\documentclass[12pt]{article}

\newcommand{\be}{\begin{equation}}
\newcommand{\ee}{\end{equation}}

\usepackage{amsmath,amsfonts,amsthm,amssymb,eucal}
\usepackage{breqn}

\newtheorem{Definition}{Definition}
\newtheorem{Theorem}{Theorem}

\newtheorem{Example}{Example}
\begin{document}

%%%%%%%%%%%%%%%%%%%%%%%%%%%%%%%%%%%%%%%%%%%%%%%%%%%%%%%%%%%
\begin{center}

Mathematics. 2021. Vol.9. No.13. Article number: 1464.

\vskip 7mm

{\bf \Large General Fractional Dynamics}

\vskip 7mm
{\bf \large Vasily E. Tarasov}$^{1,2}$ \\
\vskip 3mm

${}^1$ {\it Skobeltsyn Institute of Nuclear Physics, \\ 
Lomonosov Moscow State University, Moscow 119991, Russia} \\
{E-mail: tarasov@theory.sinp.msu.ru} \\

${}^2$ {\it Faculty "Information Technologies and Applied Mathematics", \\ 
Moscow Aviation Institute (National Research University), Moscow 125993, Russia } \\

\begin{abstract}
General fractional dynamics (GFDynamics) can be viewed as an interdisciplinary science, in which the non-local properties of linear and nonlinear dynamical systems are studied by using of general fractional calculus, equations with general fractional integrals (GFI) and derivatives (GFD), or general nonlocal mappings with discrete time. 
The GFDynamics implies research and obtaining results concerning of general form of nonlocality, which can be described by 
general form operator kernels, and not its particular implementations and representations. 

In this paper, it is proposed the concept of "general nonlocal maps" that are exact solutions of equations with GFI and GFD at discrete points. 
In these maps, the non-locality is determined by the kernels that are associated to the Sonin and Luchko kernels of general fractional integrals and derivatives, which are used in initial equations. 

Using general fractional calculus, we consider fractional systems with general non-locality in time, which are described by equations with general fractional operators and periodic kicks. 
Equations with GFI and GFD of arbitrary order are also used to derive general nonlocal maps.
Exact solutions for these general fractional differential and integral equations with kicks are obtained.
These exact solutions with discrete time points are used to derive general nonlocal maps without approximations. 
%%% In these maps, the non-locality is determined by 
%%%the associated kernels from Sonin and Luchko kernel pairs 
%%%of integral and integro-differential operators. 
Some examples of non-locality in time are described.
\end{abstract}

\end{center}

%%% 26A33: Fractional derivatives and integrals
%%% 34A08: Fractional differential equations
%%% 70G60: Dynamical systems methods
%%% 70Kxx: Nonlinear dynamics 
%%% 45E10: Integral equations of the convolution type
%%%45.10.Hj Perturbation and fractional calculus methods

%%%Keywords: Fractional differential equations, Fractional derivatives, Open systems

%%%%%%%%%%%%%%%%%%%%%%%%%%%%%%%%%%%%%%%%%%%%%%%%%%%%%%%%%%%%
%%%%%%%%%%%%%%%%%%%%%%%%%%%%%%%%%%%%%%%%%%%%%%%%%%%%%%%%%%%%
%%%%%%%%%%%%%%%%%%%%%%%%%%%%%%%%%%%%%%%%%%%%%%%%%%%%%%%%%%%%

\section{Introduction}

%%%%%%%%%%%%%%%%%%%%%%%%%%%%%%%%%%%%%%%%%%%%%%%%%%%%%%%%%%%%

Fractional dynamics \cite{Zaslavsky2004,Springer2010,KlafterLimMetzler2012,EdelmanMacauSanjuan2018} is an interdisciplinary science, in which the nonlocal properties of dynamical systems are studied by using methods of fractional calculus
\cite{FC1,FC2,FC3,FC4,FC5,GP2015,Handbook2019-1,Handbook2019-2}, 
integro-differential equations of non-integer orders and discrete nonlocal mappings. 
A fractional dynamical system is understood as a nonlocal system of any nature (physical, biological, economic, etc.), the state of which changes (discretely or continuously) in time. 
Fractional dynamics uses fractional differential equations and fractional discrete mappings to describe dynamical systems with non-locality in space and time 
in physics \cite{Handbook2019-4,Handbook2019-5}, 
biology \cite{Ionescu}, 
economics \cite{BOOK-MDPI-2020,BOOK-DG-2021} for example.

Processes with non-locality in time are characterized by the dependence of behavior of dynamical systems at a given moment in time on the history of its behavior on a finite time interval in the past. 
To describe this dependence can be used the integral and integro-differential operators that form a fractional calculus.
Fractional calculus is a branch of mathematical analysis that studies two types of integro-differential operators, for which generalized analogs of fundamental theorems hold, and therefore these types of operators are called fractional derivatives and fractional integrals. 
The characteristic property of fractional derivatives and integrals is nonlocality in space in time \cite{PFC5,Rules}.

In fractional calculus, nonlocality is described by the kernel of the operators, which are fractional integrals (FI) and fractional derivatives (FD) of non-integer orders.
To describe dynamical systems with various types of nonlocality in space and time, it is important to use 
integral and integro-differential operators
with various types of kernels that allow us to describe various types of nonlocality. 
Therefore it is important to have a general fractional calculus that allows us to describe nonlocality in the most general form. 

The concept of general fractional calculus (GFC) has been suggested by Anatoly N. Kochubei in his work \cite{Kochubei2011} in 2011 (see also \cite{Kochubei2019-1} and \cite{Kochubei2019-2,Kochubei2019-3}). 
In works in \cite{Kochubei2011,Kochubei2019-1} the general fractional derivatives (GFD) and general fractional integral (GFI) are defined. 
For these operators the general fundamental theorems are proved 
in \cite{Kochubei2011,Kochubei2019-1}. 
This approach to GFC is based on the concept of kernel pairs, which was proposed by the Russian mathematician 
Nikolay Ya. Sonin (1849-1915) in 1884 article \cite{Sonine-1} (see also \cite{Sonine-2}).
Note that in the mathematical literature, his name 
"Sonin" \cite{SoninMathnet} is mistakenly used in French transliteration as "N. Sonine" from french journals \cite{Sonine-1}.
%%%in work \cite{Sonine-1} (see also \cite{Sonine-2}). 
Then, very important results in constructing the GFC was derived by Yuri Luchko in 2021 \cite{Luchko2021-1,Luchko2021-2,Luchko2021-3}. In works \cite{Luchko2021-1,Luchko2021-2},
GFD and GFI of arbitrary order are suggested, and the general fundamental theorems for the GFI and GFDs are proved.
Operational calculus for equations with general fractional derivatives is proposed in \cite{Luchko2021-3}. 

%%%%%%%%%%%%%%%%%%%%%%%%%%%%%%%%%%%%%%%%%%%%%%%%%%%%%%%%%%%%

As a result, a mathematical basis was created for constructing the general fractional dynamics. 
In the framework of the general fractional dynamics, it is assumed and implied to obtain, not only and not so much, results concerning studies of specific types of operator kernels, but, first of all, general results that do not depend on specific types (particular implementations) of operator kernels.
In general fractional dynamics, all research and results should concern the general form of nonlocality, operator kernels of almost all types (all set of kernels), or a wide subset of such kernels. 

It should be noted that by now solutions of some equations with general operators have already been described, which form the basis of general fractional dynamics and can be used to describe some dynamic processes in different fields of science. 
Let us note some ordinary and particular fractional differential equations with GFD and GFI, which have already been considered by now and can be used in GFDynamics. 
A solution of the relaxation equations with GFD with respect to the time is described in \cite{Kochubei2011,Kochubei2019-2}.
The general growth equation with GFD, which can be used in macroeconomic models for processes with memory and distributed lag, is described in \cite{Kochubei2019-3}. 
The equations with GFD are considered in \cite{Kochubei2017,Sin2018}. 
The time-fractional diffusion equations with the GFD were described in \cite{Kochubei2011,Kochubei2019-2,LuchkoYamamoto2016,LuchkoYamamoto2020}. 
To solve general fractional differential equation, a general operational calculus was proposed in \cite{Luchko2021-3}.
The integral equations of the first kind with Sonin kernels, and the GFI and GFD of the Liouville and Marchaud type are described in \cite{Samko1,Samko2}. 
The partial differential equations containing the GFD and GFI are considered in 
\cite{KinashJanno2019-1,KinashJanno2019-2}. 
Some applications of GFC are described in recent published works (see \cite{LuchkoYamamoto2016,LuchkoYamamoto2020,Sin2018,Hanyga,Giusti} and references therein). 

%%%%%%%%%%%%%%%%%%%%%%%%%%%%%%%%%%%%%%%%%%%%%%%%%%%%%%%%%

In nonlinear dynamics, discrete-time description is usually derived from differential equations with integer-order derivatives and periodic kicks (Section 5 in \cite{SUZ},
Sections 5.2, 5.3 in \cite[pp.60-68]{Zaslavsky2}, and Section 1.2 in \cite[pp.16-17]{Schuster}, Chapter 18 in \cite[pp.409-453]{Springer2010}). 
This description is represented by mappings, in which the value $X_{n+1}$ is determined by the value $X_{n}$ 
(or a fixed finet number of values $X_n$ and $X_{n-1}$, for example), i.e. 
$X_{n+1}={\cal F}(X_n)$ (or $X_{n+1}={\cal F}(X_n,X_{n-1})$, for example). 
These maps cannot describe nonlocal dynamical systems, since
since the original equations contain only derivatives of integer order, and thus are local. 

We should note that discrete general fractional calculus 
has not yet been created. 
However, it should be noted that discrete GFDynamics can be described by using the proposed concept of general nonlocal maps.
Such discrete nonlocal mappings must be described by kernel of operators that are used in general fractional calculus. 
These maps can be derived from equations with GFD and GFI without approximations. 
In fact, these general nonlocal maps are exact solutions of fractional differential equations at discrete points. 
The concept of "general nonlocal maps" is based on the approach that was proposed for discrete fractional dynamics in
\cite{Tarasov-Zaslavsky,Tarasov-Map1,Tarasov-Map2,Springer2010}. 

In the discrete fractional dynamics, the nonlocality in time is taken into account by the maps, in which the value
$X_{n+1}$ is determined by all past values, $X_{n+1}={\cal F}(X_n,X_{n-1}, \dots ,X_1)$, where
the number of variables in the function ${\cal F}$ increases with number $n\in \mathbb{N}$. 
For the first time such a discrete mappings was derived from fractional differential equations in \cite{Tarasov-Zaslavsky}. 
Then, this approach was developed in \cite{Tarasov-Map1,Tarasov-Map2,Springer2010}, 
and it has been has been applied in 
\cite{Tarasov-Map3,TT-Logistic,Entropy2021,MMAS2021,BOOK-DG-2021,Mathematics2021}. 
We should emphasize that the nonlocal mappings are derived from fractional differential equations without approximations.
The first computer simulations of these nonlocal mappings are proposed in \cite{TT-Edelman1,TT-Edelman2}. 
Then new types of chaotic behavior of systems with nonlocality in time were discovered 
\cite{Edelman1,Edelman2,Edelman3,EdelmanTaieb,Edelman4,Edelman5},
\cite{Edelman6,Edelman7,Edelman2015,Edelman8,Edelman2018,Edelman2018-2}
\cite{Edelman-Handbook-2,Edelman-Handbook-4,Edelman2021}.
The discrete fractional calculus \cite{GP2015} are also used to derive nonlocal mappings in 
\cite{Edelman1} -- \cite{Edelman2021}.
All these mappings are described discrete fractional dynamics with power-law nonlocality in time only. 

Therefore it is important to derive general fractional dynamics, where we take into account general form of nonlocality in time.
The general nonlocal maps should be derived from equations with GFD and GFI.
These nonlocal maps can demonstrate new types of attractors and chaotic behavior of nonlocal systems.
We should emphasize that these discrete GFDynamics can be derived from equations with GFD and GFI without approximations.

In this paper, using general fractional calculus, we consider fractional systems with non-locality in time, which are described by equations with general fractional derivatives, integrals and periodic kicks. 
Exact solutions for these nonlinear fractional differential and integral equations with kicks are obtained.
These exact solutions for discrete time points are used to derive maps with non-locality in time are described without approximations. 
The non-locality of geberal nonlocal maps is determined by kernels that are associated Sonin and Luchko kernels of GFI and GFD that are used in initisl equations.

%%%%%%%%%%%%%%%%%%%%%%%%%%%%%%%%%%%%%%%%%%%%%%%%%%%%%%%%%

\section{Equations of General Fractional Dynamical Systems}

Fractional dynamics with continuous time is described by the integral and integro-differential equations. 
To take into account non-locality in time, we should require that these equations cannot be represented as differential equations or systems of differential equations of integer orders only.

To described dynamical systems with general form of non-locality in time, it is possible to use the following equations
\be \label{EQ-I}
I^{t}_{(M)} [\tau] \, X(\tau) = F_I(t,X(t)) ,
\ee
\be \label{EQ-D}
D^{t}_{(K)} [\tau] \, X(\tau) = F_D(t,X(t)) ,
\ee
\be \label{EQ-Dc}
D^{t,*}_{(K)} [\tau] \, X(\tau) = F_D(t,X(t)) ,
\ee
where the operators $I^{t}_{(M)}$ and $D^{t}_{(K)}$ have the forms
\be \label{Oper-IM}
I^{t}_{(M)} [\tau] \, X(\tau) = (M \, * \, X)(t)=
\int^t_{0} \, d\tau \, M(t -\tau) \, X (\tau) 
\ee
\be \label{Oper-DK}
D^{t}_{(K)} [\tau] \, X(\tau) = 
\frac{d}{dt} (K \, * \, X)(t)=
\frac{d}{dt} \int^t_{0} \, d\tau \, K(t - \tau) \, X (\tau) 
\ee
\be \label{Oper-DKc}
D^{t,*}_{(K)} [\tau] \, X(\tau) = (K \, * \, X^{(1)})(t)=
\int^t_{0} \, d\tau \, K(t -\tau) \, X^{(1)} (\tau) 
\ee
In this case, the non-locality in time is characterized by the kernels $M(t,\tau)$, $K(t -\tau)$. The operators
\eqref{Oper-IM}, \eqref{Oper-DK} and \eqref{Oper-DKc} are defined through the Laplace convolution $*$ with these kernels.

%%% Example 1.
\begin{Example} 
For the integral kernel
\be \label{M-PL}
M(t-\tau) = \frac{(t-\tau)^{\alpha-1}}{\Gamma(\alpha)}, 
\ee
equation \eqref{EQ-I} gives the equations 
\be \label{EQ-M-PL}
I^{\alpha}_{RL} [\tau] X(\tau) = F_I(t,X(t)) ,
\ee
where $I^{\alpha}_{RL}$ is the Riemann-Liouville integral of the order $\alpha \in (0,\infty)$, \cite{FC4}.
\end{Example}

%%% Example 2.
\begin{Example} 
For the kernel
\be \label{K-PL}
K(t-\tau) = \frac{(t-\tau)^{-\alpha}}{\Gamma(1-\alpha)}, 
\ee
equations \eqref{EQ-D} and \eqref{EQ-Dc} give the equations 
\be \label{EQ-K-PL}
D^{\alpha}_{RL} [\tau] X(\tau) = F_D(t,X(t)) ,
\quad
D^{\alpha}_{C} [\tau] X(\tau) = F_D(t,X(t)) ,
\ee
where $D^{\alpha}_{RL}$ and $D^{\alpha}_{C}$ are the Riemann-Liouville and Caputo fractional derivatives of the order $\alpha (0,1)$, respectively, \cite[pp.70,92]{FC4}.
\end{Example} 

Equations \eqref{EQ-M-PL} and \eqref{EQ-K-PL} 
describe dynamical systems with power-law nonlocality in time, which is interpreted as a fading memory and kernels are called the memory function.

Kernels \eqref{M-PL} and \eqref{K-PL} can be considered as an approximation of a more general form of kernels \cite{General}
to describe dynamical systems with nonlocality in time. 
In article \cite{General}, using the generalization of the Taylor series that is proposed in \cite{TRB} for the kernels, 
we proved that for the wide class of kernels can be approximately considered as a power function.

For applications, it is important to describe dynamical systems with a more general form of nonlocality in time. 
Therefore, it is necessary to consider the general type of kernels for integral and integro-differential operators 
\eqref{Oper-IM}, \eqref{Oper-DK} and \eqref{Oper-DKc}.

To do this, we must have a general integral operators and integro-differential operator, which can be interpretedas a generalization of the standardintegrals and derivative of integer order. Moreover, these general fractional operators must form a certain calculus, and the fundamental theorems of this calculus must hold for these operators. 

The general fractional calculus (GFC) is a branch of mathematical analysis of the integral and integro-differential operators that are generalizations of integrals and derivatives of integer orders, for which the generalization of the fundamental theorems of calculus are satisfied.

The GFC is based on the concept of a pair of mutually associated Sonin's kernels (Sonine kernels). 
This type of kernels was proposed by the Russian mathematician 
Nikolay Ya. Sonin (1849-1915) in 1884 \cite{Sonine-1,Sonine-2}, publishing articles in French as "N. Sonine".

The concept of mutual associated kernels of two operators, which are defined through the Laplace convolution, assumes that the Laplace convolution for the kernels of these operators is equal to one. For the operators \eqref{Oper-IM}, \eqref{Oper-DK} and \eqref{Oper-DKc}, the Sonin's condition for the kenels
$M(t)$ and $K(t)$ requires that the relations 
\be \label{Sonin-Cond}
\int^t_0 M(t-\tau) \, K(\tau) \, d\tau = 1 
\ee
holds for all $t \in (0,\infty)$.

In this work about GFDynamics, we will follow the Luchko's approach to general fractional calculus that is proposed by Yuri Luchko in articles \cite{ Luchko2021-1, Luchko2021-2} in 2021.

\begin{Definition} 
The functions $M(t), K(t) \in C_{-1,0}(0,\infty)$ are a Sonin's pair of kernels (or the Sonin kernels), if
\be
M(t), K(t) \in C_{-1,0}(0,\infty) 
\ee 
and the Sonin condition \eqref{Sonin-Cond} holds, where
\begin{equation} \label{Cabe}
C_{a,b}(0,\infty) \, = \, 
\{ X(t): \ X(t) = t^{p} \, Y(t),\ t>0,\ 
a < p < b, \ Y(t) \in C[0,\infty) \} .
\end{equation}
The set of such Sonin's kernels is denoted by $\mathcal{S}_{-1}$, and $\left( M(t), \, K(t)\right) \in \mathcal{S}_{-1}$, if
\begin{equation} \label{Son-1}
M(t), \, K(t) \in C_{-1,0}(0,\infty) \quad \text{and} \quad
(M \, * \, K)(t) = 1 
\end{equation}
for all $t \in t \in (0,\infty)$. 
\end{Definition}

\begin{Definition} \label{DEFIN-GFI}
Let $M(t) \in \mathcal{S}_{-1}$ and 
$X(t) \in C_{-1}(0,\infty) = C_{-1,\infty}(0,\infty)$. 
The general fractional integral with the kernel 
$M(t) \in C_{-1,0}(0,\infty)$ is the operator on the space $C_{-1}(0,\infty)$, that is 
\begin{equation} 
I^{t}_{(M)}: \ C_{-1}(0,\infty)\, \rightarrow C_{-1}(0,\infty),
\end{equation}
that is defined by the equation
\begin{equation} \label{DEF-GFI}
I^{t}_{(M)} [\tau] X(\tau) = 
(M \, * \, X)(t) = 
\int^t_0 d \tau \, M(t-\tau) \, X(\tau) .
\end{equation}
\end{Definition}

The following important property allows to define the repeated action general operator \eqref{DEF-GFI}.
Let $M_1(t), M_2(t)\in \mathcal{S}_{-1}$ and 
$X(t) \in C_{-1}(0,\infty)$.
Then the equality 
\begin{equation} 
I^{t}_{(M_1)} [\tau]\, I^{\tau}_{(M_2)}[s] X(s) = 
I^{t}_{(M_1*M_2)}[\tau] X(\tau) .
\end{equation}
holds for $t>0$.

Let us assume that the kernels $M(t)$ and $K(t)$ form are a Sonin's pair of kernels. This allows us to defined general fractional derivatives $D^{t}_{(K)}$ and $D^{t,*}_{(K)}$ 
that are associated with general fractional integral $I^{t}_{(M)}$.

The space $C^m_{-1}(0,\infty) \subset C_{-1}(0,\infty)$, where $m \in \mathbb{N}$, consists of functions $X(t)$, for which 
$X^{(m)} \in C_{-1}(0,\infty)$.

\begin{Definition} \label{DEFIN-GFD}
Let $K(t) \in \mathcal{S}_{-1}$ and 
$X(t) \in C^1_{-1}(0,\infty)$.
The general fractional derivative of the Riemann-Liouville type
with kernel $K(t) \in C_{-1,0}(0,\infty)$, which is associated with GFI \eqref{DEF-GFI} is defined as 
\begin{equation} \label{DEF-GFD-RL}
D^{t}_{(K)} [\tau] X(\tau) = 
\frac{d}{dt} (K \, * \, X )(t) =
\frac{d}{dt} \int^{t}_{0} d\tau \, K(t-\tau) \, X(\tau) .
\end{equation}
for $t \in (0,\infty)$.

The general fractional derivative of the Caputo type 
with kernel $K(t) \in C_{-1,0}(0,\infty)$ is defined as 
\begin{equation} \label{DEF-GFD-C}
D^{t,*}_{(K)} [\tau] X(\tau) = 
(K \, * \, X^{(1)})(t) =
\int^{t}_{0} d\tau \, K(t-\tau) \, X^{(1)}(\tau) 
\end{equation}
for $t \in (0,\infty)$.
\end{Definition}

These GFI and GFD satisfy the fundamental teorems of GFC \cite{Luchko2021-1,Luchko2021-2}.

%%%%%%%%%%%%%%%%%%%%%%%%%%%%%%%%%%%%%%%%%%%%%%%%%%%%%%%%%%

\section{To General Fractional Dynamics}

General fractional dynamics is an interdisciplinary science, in which the non-local properties of nonlinear dynamical systems are studied. General fractional dynamics uses linear and nonlinear models to describe systems with general forms of non-locality in space and time by using equations with operators of general fractional calculus. 

We can conditionally distinguish the following three directions in the general fractional dynamics with non-locality in time. 

\begin{itemize}

\item
1) General fractional dynamics with continuous time is described by the equations with GFD and GFI with the Sonin kernels.
For example, solutions of the integral equations with Sonin kernels 
\be
\int^t_{-\infty} d\tau \, M(t-\tau) \, X(\tau) = F(t)
\ee
in which the lower limit is not zero but minus infinity, were proposed Stefan G. Samko and Rogerio P. Cardoso \cite{Samko1,Samko2}, where general fractional integrals and derivatives are considered in the 
Marchaud-type and Liouville type.

For example, solution of the relaxation equations 
\be \label{EQ-Dc-Lin}
D^{t,*}_{(K)} [\tau] \, X(\tau) = \lambda X(t) ,
\ee
where relaxation means that $\lambda<0$, 
was derived by Anatoly N. Kochubei in works \cite{Kochubei2011,Kochubei2019-2}. 
Solution of the growth equation that is described by \eqref{EQ-Dc-Lin} with $\lambda>0$ is derived by Anatoly N. Kochubei and Yuri Kondratiev \cite{Kochubei2019-3} in 2019. 
Operational calculus for equations with general fractional derivatives with the Sonin kernels is proposed in
\cite{Luchko2021-3}. 

\item
2) General fractional dynamics with discrete time could be described by a discrete analogue of general fractional calculus. However, such a "discrete general fractional calculus" has not yet been created.

\item
3) Another way of describing general fractional dynamics which discrete time can be based on the use of discrete mappings obtained from exact solutions of general fractional differential and integral equations with periodic kicks. 
For the first time discrete mappings with non-locality in time were derived from fractional differential equations in papers 
\cite{Tarasov-Zaslavsky,Tarasov-Map1,Tarasov-Map2}, in which the Riemann-Liouville and Caputo fractional derivatives were used. 
(see also \cite{Springer2010} and \cite{Tarasov-Map3,TT-Logistic,Entropy2021,MMAS2021}). 
The proposed approach approach allows us to derive discrete-time mappings with non-locality in time from integro-differential equations of non-integer orders without approximation. 
This approach is a generalization of methods that is well-known in nonlinear dynamics and theory of chaos
(Sec. 5 in \cite{SUZ},
Sec. 5.2, 5.3 in \cite[pp.60-68]{Zaslavsky2}, and 
Sec. 1.2 in \cite[pp.16-17]{Schuster}, 
Ch. 18 in \cite[pp.409-453]{Springer2010}), 
where the discrete-time dynamical mapping are derived from ordinary differential equations of integer orders with kicks. 

\end{itemize}

Let us consider equations \eqref{EQ-I}, \eqref{EQ-D}, \eqref{EQ-Dc} with GFD and GFI, and periodic kicks, by using
\be \label{F-N-delta}
F_I(t,X(t)) = F_D (t,X(t)) = 
\lambda \, {\cal G}(t,X(t)) 
\sum^{\infty}_{k=1} \delta ((t+\epsilon)/T-k) ,
\ee
where $T$ is period of the periodic sequence of kicks, 
$\lambda$ is an amplitude of these kicks, 
${\cal G}(t,X)$ is a real-valued function.

Expression \eqref{F-N-delta} contains the delta functions that are distributions (the generalized functions). 
Therefore this expression and equations \eqref{EQ-I}, \eqref{EQ-D}, \eqref{EQ-Dc} with function \eqref{F-N-delta}
are treated as continuous functionals on a space of test functions (see Sec 8 \cite[pp.145-160]{FC1} and \cite{GenFunc1,GenFunc2}).

We also use $t-\epsilon$, where $0<\epsilon \ll T$, instead of $t$ in the argument of the delta functions, to make a sense of the product of ${\cal G}(t,X(t))$
and the delta function for the time points, where $X(t-0) \neq X(t+0)$, \cite{Edelman6}. 

We will also assume that the function ${\cal G}(t,X(t))$ is such that the product of the delta function is defined in the neighborhood of the points $t=Tk-\epsilon$, where $k \in \mathbb{N}$
and $\epsilon \to 0$. 

We start first by considering the integro-differential equation \eqref{EQ-Dc} with the GFD $D^{t,*}_{(K)}$ with the Sonin kernal $K(t) \in C_{-1,0}(0,\infty)$ and periodic kicks \eqref{F-N-delta}.

\begin{Theorem} \label{Theorem-1}
Let $K(t) \in C_{-1,0}(0,\infty)$ and 
$X(t) \in C^1_{-1}(0,\infty)$.
Then the integro-differential equation
\be \label{EQ-Th1-1}
D^{t,*}_{(K)} [\tau] \, X(\tau) = \lambda \, {\cal G}(t,X(t)) 
\sum^{\infty}_{k=1} \delta ((t+\epsilon)/T-k) ,
\ee
has the solution 
\be \label{EQ-Th1-2}
X(t) = X(0) \, + \, \lambda \, T \, \sum^{n}_{k=1} 
M(t-(Tk-\epsilon)) \, {\cal G}(Tk-\epsilon,X(Tk-\epsilon)) ,
\ee
if $Tn<t<T(n+1)$, where $M(t) \in C_{-1,0}(0,\infty)$ is a function that is associated kernel to the kernel $K(t)$, i.e. 
the functions $K(t)$, $M(t)$ form a mutually associated pair of Sonin's kernels.
\end{Theorem}

\begin{proof} 
Let us use the general fractional integral
\begin{equation} \label{Proof-Th1-1}
I^{s}_{(M)} [\tau] X(\tau) = (M \, * \, X)(t) = 
\int^t_0 d \tau \, M(t-\tau) \, X(\tau) .
\end{equation}
with the kernel $M(t) \in C_{-1,0}(0,\infty)$ that is associated with the kernel $K(t)$, such that $K(t),M(t)$ form a mutually associated pair of Sonin's kernels.
Applying integral operator \eqref{Proof-Th1-1}, to equation \eqref{EQ-Th1-1}, we obtain 
\be \label{Proof-Th1-2}
I^{s}_{(M)} [t]
D^{t,*}_{(K)} [\tau] \, X(\tau) = \lambda \, 
I^{s}_{(M)} [t] \, {\cal G}(t,X(t)) 
\sum^{\infty}_{k=1} \delta ((t+\epsilon)/T-k) ,
\ee
where $s>t>\tau>0$.
Using the second fundamental theorem of general fractional operators in the form of equality
\be
I^{s}_{(M)} [t] D^{t,*}_{(K)} [\tau] \, X(\tau) = 
X(s) \, - \, X(0)
\ee
that holds for $X(t) \in C^1_{-1}(0,\infty)$ (see 
Eq. 60 of Theorem 4 in \cite[p.11]{ Luchko2021-1},
Eq. 33 of Theorem 2 in \cite[p.7]{ Luchko2021-2}), we derive
\be \label{Proof-Th1-3a}
X(s) \, - \, X(0)= \lambda \, 
I^{s}_{(M)} [t] \, {\cal G}(t,X(t)) 
\sum^{\infty}_{k=1} \delta ((t+\epsilon)/T-k) .
\ee
Using \eqref{Proof-Th1-1} equation \eqref{Proof-Th1-3a} is written as
\be \label{Proof-Th1-3}
X(s) \, - \, X(0) = \lambda \, 
\int^s_0 d t \, M(s-t) \, {\cal G}(t,X(t)) 
\sum^{\infty}_{k=1} \delta ((t+\epsilon)/T-k) .
\ee
For $Tn<s<T(n+1)$, equation \eqref{Proof-Th1-3} is 
\be \label{Proof-Th1-4}
X(s) \, - \, X(0)= \lambda \, \sum^{n}_{k=1}
\int^s_0 d t \, M(s-t) \, {\cal G}(t,X(t)) \delta ((t+\epsilon)/T-k) .
\ee
Using the equality
\be \label{Proof-delta}
\int^s_0 d t f(t) \delta (t-a) = f(a)
\ee
which holds for and $0<a<s$, and $f(t) \in C^{\infty}(\Omega_a)$,  
where $\Omega_a$ is a neighborhood of point $t=a$, 
equation \eqref{Proof-Th1-4} gives
\be \label{Proof-Th1-5}
X(s) \, - \, X(0)= \lambda \, T \, \sum^{n}_{k=1}
M(s- (kT-\epsilon) ) \, {\cal G}(kT-\epsilon,X((kT-\epsilon) )) .
\ee
Equation \eqref{Proof-Th1-5} leads to solution \eqref{EQ-Th1-2}.

\end{proof}

Let us consider solution for the discrete time points $t=Tk$ with $k \in \mathbb{N}$.

\begin{Theorem} \label{Theorem-2}
Let $K(t) \in C_{-1,0}(0,\infty)$ and 
$X(t) \in C^1_{-1}(0,\infty)$.
The solution of the equation
\be \label{EQ-Th2-1}
D^{t,*}_{(K)} [\tau] \, X(\tau) = \lambda \, {\cal G}(t,X(t)) 
\sum^{\infty}_{k=1} \delta ((t+\epsilon)/T-k) ,
\ee
for the time points $s = Tk-\epsilon$ at $\epsilon \to 0+$, and 
the variables 
\begin{equation} \label{EQ-Th2-2}
X_{k}=\lim_{\epsilon \rightarrow 0+} X(Tk-\epsilon),
\quad (k=1,...,n+1)
\end{equation} 
is the non-local mapping
\be \label{EQ-Th2-3}
X_{n+1}\, = \, X_n + \lambda \, T \, M(T) \, {\cal G}(Tn,X_n) + 
\lambda \, T \, \sum^{n-1}_{k=1} 
\Bigl( M(T(n+1-k)) - M(T(n-k)) \Bigr) \, {\cal G}(Tk,X_k) .
\ee
where $M(t) \in C_{-1,0}(0,\infty)$ is a function that is associated kernel to the kernel $K(t)$, i.e. 
the functions $K(t)$, $M(t)$ form a mutually associated pair of Sonin's kernels.
\end{Theorem}

\begin{proof} 
Solution \eqref{EQ-Th1-2}, which is derived in Theorem \ref{Theorem-1}, for $t= T(n+1)-\epsilon $ and $t=Tn-\epsilon$
with $0<\epsilon\ll T$ is given by the equations 
\be \label{Proof-Th2-1}
X(T(n+1)-\epsilon) = X(0) \, + \, \lambda \, \sum^{n}_{k=1} 
M(T(n+1)-Tk) \, {\cal G}(Tk-\epsilon,X(Tk-\epsilon)) ,
\ee
\be \label{Proof-Th2-2}
X(Tn-\epsilon) = X(0) \, + \, \lambda \, T \, \sum^{n-1}_{k=1} 
M(Tn-Tk) \, {\cal G}(Tk-\epsilon,X(Tk-\epsilon)) .
\ee
Using the variables 
\begin{equation} \label{Proof-Th2-3}
X_{k}=\lim_{\epsilon \rightarrow 0+} X(Tk-\epsilon),
\quad (k=1,...,n+1) ,
\end{equation} 
solutions \eqref{Proof-Th2-1} and \eqref{Proof-Th2-2} 
at the limit $\epsilon \to 0+$ give
\be \label{Proof-Th2-4}
X_{n+1} = X(0) \, + \, \lambda \, T \, \sum^{n}_{k=1} 
M(T(n+1-k)) \, {\cal G}(Tk,X_k) ,
\ee
\be \label{Proof-Th2-5}
X_n = X(0) \, + \, \lambda \, T \, \sum^{n-1}_{k=1} 
M(T(n-k)) \, {\cal G}(Tk,X_k) .
\ee
Subtraction equation \eqref{Proof-Th2-5} from equation \eqref{Proof-Th2-4} gives
\be \label{Proof-Th2-6}
X_{n+1}\, - \, X_n = \lambda \, T \, M(T) \, {\cal G}(Tn,X_n) + 
\lambda \, T \, \sum^{n-1}_{k=1} 
\Bigl( M(T(n+1-k)) - M(T(n-k)) \Bigr) \, {\cal G}(Tk,X_k) .
\ee
Equation \eqref{Proof-Th2-6} leads to \eqref{EQ-Th2-3}.

\end{proof}

Let us consider equation with GFD of the Riemann-Liouville type.

\begin{Theorem} \label{Theorem-3}
Let $K(t) \in C_{-1,0}(0,\infty)$ and 
$X(t) \in C^1_{-1}(0,\infty)$.
The equation
\be \label{EQ-Th3-1}
D^{t}_{(K)} [\tau] \, X(\tau) = \lambda \, {\cal G}(t,X(t)) 
\sum^{\infty}_{k=1} \delta ((t+\epsilon)/T-k) ,
\ee
has the solution 
\be \label{EQ-Th3-2}
X(t) \, = \, \lambda \, T \, \sum^{n}_{k=1} 
M(t-(Tk-\epsilon)) \, {\cal G}(Tk-\epsilon,X(Tk-\epsilon)) ,
\ee
if $Tn<t<T(n+1)$.
For the time points $s = Tk-\epsilon$ at $\epsilon \to 0+$, and 
the variables 
\begin{equation} \label{EQ-Th3-Xk}
X_{k}=\lim_{\epsilon \rightarrow 0+} X(Tk-\epsilon),
\quad (k=1,...,n+1) ,
\end{equation} 
solution \eqref{EQ-Th3-2} is represented by the nonlocal mapping
\[
X_{n+1} = X_n \,+ \, \lambda \, T \, M(T) \, {\cal G}(T(n+1),X_n) +
\]
\be \label{EQ-Th3-3}
\lambda \, T \, \sum^{n-1}_{k=1} 
\left( M( T(n+1-k) ) - M( T(n-k) ) \right) \, {\cal G}(kT,X_k) .
\ee
where $M(t) \in C_{-1,0}(0,\infty)$ is a function that is associated kernel to the kernel $K(t)$, i.e. 
the functions $K(t)$, $M(t)$ form a mutually associated pair 
of Sonin's kernels.
\end{Theorem}

\begin{proof} 
The proof of this Theorem is similar to the proofs of Theorems \ref{Theorem-1} and \ref{Theorem-2}. 
In this, proof we use the second fundamental theorem of general fractional operators in the form of equality
\be \label{Proof-Th3-0}
I^{s}_{(M)} [t] D^{t}_{(K)} [\tau] \, X(\tau) = X(s) 
\ee
that holds for $X(t) \in C^1_{-1}(0,\infty)$ (see 
Eq. 61 of Theorem 4 in \cite[p.11]{ Luchko2021-1},
Eq. 34 of Theorem 2 in \cite[p.7]{ Luchko2021-2}).

%%% Part 1
Applying the GFI $I^{s}_{(M)} [t]$ with the kernal $M(t)$, which is associated to the kernel ${\cal G}(t)$ of the GFD $D^{t}_{(K)}$,
to equation \eqref{EQ-Th3-1}, and using \eqref{Proof-Th3-0}, 
we derive
\be \label{Proof-Th3-1a}
X(s) \, = \, \lambda \, 
I^{s}_{(M)} [t] \, {\cal G}(t,X(t)) 
\sum^{\infty}_{k=1} \delta ((t+\epsilon)/T-k) .
\ee
Using \eqref{Proof-Th1-1}, equation \eqref{Proof-Th3-1a} is written as
\be \label{Proof-Th3-2a}
X(s) \, = \, \lambda \, 
\int^s_0 d t \, M(s-t) \, {\cal G}(t,X(t)) 
\sum^{\infty}_{k=1} \delta ((t+\epsilon)/T-k) .
\ee
For $Tn<s<T(n+1)$, equation \eqref{Proof-Th3-2a} is 
\be \label{Proof-Th3-3a}
X(s) \, = \, \lambda \, \sum^{n}_{k=1}
\int^s_0 d t \, M(s-t) \, {\cal G}(t,X(t)) \delta ((t+\epsilon)/T-k) .
\ee
Using the equality \eqref{Proof-delta}, equation \eqref{Proof-Th3-3a} gives
\be \label{Proof-Th3-5a}
X(s) \, = \, \lambda \, T \, \sum^{n}_{k=1}
M(s- (kT-\epsilon) ) \, {\cal G}(kT-\epsilon,X((kT-\epsilon) )) .
\ee

%%% Part 2
Using \eqref{Proof-Th3-5a} for $s= T(n+1)-\epsilon$ and
$s= Tn-\epsilon$, weobtain
\be \label{Proof-Th3-1b}
X(T(n+1)-\epsilon) \, = \, \lambda \, T \, \sum^{n}_{k=1} 
M(T(n+1)-Tk) \, {\cal G}(Tk-\epsilon,X(Tk-\epsilon)) ,
\ee
\be \label{Proof-Th3-2b}
X(Tn-\epsilon) \, = \, \lambda \, T \, \sum^{n-1}_{k=1} 
M(Tn-Tk) \, {\cal G}(Tk-\epsilon,X(Tk-\epsilon)) .
\ee
Using variables \eqref{EQ-Th3-Xk}, solutions \eqref{Proof-Th3-1b} and \eqref{Proof-Th3-2b} at the limit $\epsilon \to 0+$ give
\be \label{Proof-Th3-1c}
X_{n+1} = \, \lambda \, T \, \sum^{n}_{k=1} 
M(T(n+1-k)) \, {\cal G}(Tk,X_k) ,
\ee 
\be \label{Proof-Th3-2c}
X_n = \, \lambda \, T \, \sum^{n-1}_{k=1} M(T(n-k)) \, {\cal G}(Tk,X_k) .
\ee 
Subtracting from equation \eqref{Proof-Th3-1c} 
equation \eqref{Proof-Th3-2c}, we obtain \eqref{EQ-Th3-3}.

\end{proof}

%%%%%%%%%%%%%%%%%%%%%%%%%%%%%%%%%%%%%%%%%%%%%%%%%%%%%%%%%%
%%%%%%%%%%%%%%%%%%%%%%%%%%%%%%%%%%%%%%%%%%%%%%%%%%%%%%%%%%
%%%%%%%%%%%%%%%%%%%%%%%%%%%%%%%%%%%%%%%%%%%%%%%%%%%%%%%%%%

Let us consider equation with general fractional integral operator and kicks.

\begin{Theorem} \label{Theorem-5}
Let $M(t) \in C_{-1,0}(0,\infty)$ and 
$X(t) \in C_{-1}(0,\infty)$.
Then the integral equation
\be \label{EQ-Th5-1}
I^{t}_{(M)} [\tau] \, X(\tau) = \lambda \, {\cal G}(t,X(t)) 
\sum^{\infty}_{k=1} \delta ((t+\epsilon)/T-k) ,
\ee
has the solution 
\be \label{EQ-Th5-2}
X(t) \, = \, \lambda \, T \, \sum^{n}_{k=1} 
K^{(1)}(t-(Tk-\epsilon)) \, {\cal G}(Tk-\epsilon,X(Tk-\epsilon)) ,
\ee
for $Tn<t<T(n+1)$, where $K^{(1)}(z) = (dK(t)/dt)_{t=z}$, 
and $K(t) \in C^1_{-1,0}(0,\infty)$ is associated kernel to the kernel $M(t)$, i.e. the functions $K(t)$, $M(t)$ form a mutually associated pair of Sonin's kernels.
\end{Theorem}

\begin{proof} 
Let us use the general fractional derivative of the Riemann-Liouville type
\begin{equation} \label{Proof-Th5-1}
D^{s}_{(K)} [t] X(t) = 
\frac{d}{ds} (K \, * \, X)(s) = 
\frac{d}{ds} \int^s_0 d t \, K(s-t) \, X(t) 
\end{equation}
with the kernel $K(t) \in C_{-1,0}(0,\infty)$ that is associated kernel to the kernel $M(t)$ of GFI, such that $K(t),M(t)$ form a mutually associated pair of Sonin's kernels.
Applying operator \eqref{Proof-Th5-1}, to equation \eqref{EQ-Th5-1}, we obtain 
\be \label{Proof-Th5-2}
D^{s}_{(K)} [t]
I^{t}_{(M)} [\tau] \, X(\tau) = \lambda \, 
D^{s}_{(K)} [t] \, {\cal G}(t,X(t)) 
\sum^{\infty}_{k=1} \delta ((t+\epsilon)/T-k) ,
\ee
where $s>t>\tau>0$.
Using the first fundamental theorem of general fractional operators in the form of equality
\be
D^{s}_{(M)} [t] I^{t}_{(K)} [\tau] \, X(\tau) = X(s) 
\ee
that holds for $X(t) \in C^1_{-1}(0,\infty)$ (see 
Eq. 51 of Theorem 3 in \cite[p.9]{Luchko2021-1},
Eq. 31 of Theorem 1 in \cite[p.6]{Luchko2021-2}), 
equation \eqref{Proof-Th5-2} takes the form
\be \label{Proof-Th5-2b}
X(s) \, = \, \lambda \, 
D^{s}_{(K)} [t] \, {\cal G}(t,X(t)) 
\sum^{\infty}_{k=1} \delta ((t+\epsilon)/T-k) ,
\ee
Then using expression \eqref{Proof-Th5-1}, we derive
\be \label{Proof-Th5-3}
X(s) \, = \, \lambda \, 
\frac{d}{ds} \int^s_0 d t \, K(s-t) \, {\cal G}(t,X(t)) 
\sum^{\infty}_{k=1} \delta ((t+\epsilon)/T-k) .
\ee
For $Tn<s<T(n+1)$, equation \eqref{Proof-Th5-3} is 
\be \label{Proof-Th5-4}
X(s) \, = \, \lambda \, \sum^{n}_{k=1}
\frac{d}{ds} \int^s_0 d t \, K(s-t) \, {\cal G}(t,X(t)) \delta ((t+\epsilon)/T-k) .
\ee
Using the equality \eqref{Proof-delta}
equation \eqref{Proof-Th5-4} gives
\be \label{Proof-Th5-5}
X(s) \, = \, \lambda \, T \, \frac{d}{ds} \sum^{n}_{k=1}
K(s- (kT-\epsilon) ) \, {\cal G}(kT-\epsilon,X((kT-\epsilon) )) .
\ee
For for $Tn<t<T(n+1)$, we obtain
\be \label{Proof-Th5-6}
X(s) \, = \, \lambda \, T \, \sum^{n}_{k=1}
K^{(1)}(s- (kT-\epsilon) ) \, {\cal G}(kT-\epsilon,X((kT-\epsilon) )) .
\ee
Equation \eqref{Proof-Th5-5} leads to solution \eqref{EQ-Th5-2}.

\end{proof}

\begin{Theorem} \label{Theorem-6}
Let $M(t) \in C_{-1,0}(0,\infty)$ and 
$X(t) \in C_{-1}(0,\infty)$.
The solution of the equation
\be \label{EQ-Th6-1}
I^{t}_{(M)} [\tau] \, X(\tau) = \lambda \, {\cal G}(t,X(t)) 
\sum^{\infty}_{k=1} \delta ((t+\epsilon)/T-k) ,
\ee
for the time points $s = Tk-\epsilon$ at $\epsilon \to 0+$, and 
the variables 
\begin{equation} \label{EQ-Th6-2}
X_{k}=\lim_{\epsilon \rightarrow 0+} X(Tk-\epsilon),
\quad (k=1,...,n+1)
\end{equation} 
is represented as the nonlocal mapping
\[
X_{n+1}\, = \, X_n + \, \lambda \, T \, K^{(1)}(T) \, {\cal G}(Tn,X_n) 
+ 
\]
\be \label{EQ-Th6-3}
\, \lambda \, T \, \sum^{n}_{k=1} 
\Bigl( K^{(1)}(T(n+1-k)) - K^{(1)}(T(n-k)) \Bigr)\, {\cal G}(Tk,X_k) ,
\ee
where $K^{(1)}(z) = (dK(t)/dt)_{t=z}$, 
and $K(t) \in C^1_{-1,0}(0,\infty)$ is a function that is associated kernel to the kernel $M(t)$, i.e. 
the functions $K(t)$, $M(t)$ form a mutually associated pair of Sonin's kernels.
\end{Theorem}

\begin{proof} 
Solution \eqref{EQ-Th5-2}, which is derived in Theorem \ref{Theorem-5} for $t= T(n+1)-\epsilon $ and $t=Tn-\epsilon$
with $0<\epsilon\ll T$ is given by the equations 
\be \label{Proof-Th6-1}
X(T(n+1)-\epsilon) \, = \, \lambda \, T \, \sum^{n}_{k=1} 
K^{(1)}(T(n+1)-Tk) \, {\cal G}(Tk-\epsilon,X(Tk-\epsilon)) ,
\ee
\be \label{Proof-Th6-2}
X(Tn-\epsilon) \, = \, \lambda \, T \, \sum^{n-1}_{k=1} 
K^{(1)}(Tn-Tk) \, {\cal G}(Tk-\epsilon,X(Tk-\epsilon)) ,
\ee
Using the variables \eqref{EQ-Th6-2}, equations
\eqref{Proof-Th6-1} and \eqref{Proof-Th6-2} 
at the limit $\epsilon \to 0+$, give
\be \label{Proof-Th6-4}
X_{n+1} \, = \, \lambda \, T \, \sum^{n}_{k=1} 
K^{(1)}(T(n+1-k)) \, {\cal G}(Tk,X_k) ,
\ee
\be \label{Proof-Th6-5}
X_{n} \, = \, \lambda \, T \, \sum^{n-1}_{k=1} 
K^{(1)}(T(n-k)) \, {\cal G}(Tk,X_k) ,
\ee
Subtraction of equation \eqref{Proof-Th2-5} from equation \eqref{Proof-Th2-4} gives
\[
X_{n+1}\, - \, X_n = \, \lambda \, T \, K^{(1)}(T) \, {\cal G}(Tn,X_n) 
+ 
\]
\be \label{Proof-Th6-6}
\, \lambda \, T \, \sum^{n-1}_{k=1} 
\left(K^{(1)}(T(n+1-k)) -K^{(1)}(T(n-k)) \right)\, {\cal G}(Tk,X_k) .
\ee
Equation \eqref{Proof-Th6-6} gives \eqref{EQ-Th6-3}.

\end{proof} 

%%%%%%%%%%%%%%%%%%%%%%%%%%%%%%%%%%%%%%%%%%%%%%%%%%%%%%%%%%
%%%%%%%%%%%%%%%%%%%%%%%%%%%%%%%%%%%%%%%%%%%%%%%%%%%%%%%%%%
%%%%%%%%%%%%%%%%%%%%%%%%%%%%%%%%%%%%%%%%%%%%%%%%%%%%%%%%%%

\section{Examples Non-Locality in form of Sonin Kernels}

In this section, examples of the pair of the Sonin kernels are presented.

Note that if the kernel $M(t)$ is associated to kernel $K(t)$,  then the kernel $K(t)$ is associated to $M(t)$. 
Therefore, if we have the operators
\begin{equation} \label{MK-KM}
I^{t}_{(M)} [\tau] X(\tau) = (M \, * \, X)(t), \quad
D^{t,*}_{(K)} [\tau] X(\tau) = (K \, * \, X^{(1)})(t) ,
\end{equation}
where $M(t), K(t) \in \mathcal{S}_{-1}$, then we can use the operators
\begin{equation} \label{KM-MK}
I^{t}_{(K)} [\tau] X(\tau) = (K \, * \, X)(t), \quad
D^{t,*}_{(M)} [\tau] X(\tau) = (M \, * \, X^{(1)})(t) .
\end{equation}
 
Let us give examples of the kernels that satisfy the Sonin condition
\begin{equation} 
\left(M \, * \, K\right)(t) =
\int^{t}_0 M(t-\tau) \, K(\tau) \, d\tau \, = \, 1 .
\end{equation} 

%%% Example 1.
\begin{Example} 
The following kernels
\begin{equation} \label{Example-1-M}
M(t) = h_{\alpha}(t) = \frac{t^{\alpha -1}}{\Gamma(\alpha)},
\end{equation}
\begin{equation} \label{Example-1-K}
K(t) = h_{1-\alpha}(t) = \frac{t^{-\alpha}}{\Gamma(1-\alpha)},
\end{equation}
where $0< \alpha <1$, 
are mutually associated Sonin kernels.
For these kernels, the GFI and GFD are well-known Riemann-Liouville and Caputo fractional operators \cite{FC4}.
\end{Example}

%%% Example 2.
\begin{Example}
The pair of the Sonin kernels \cite[p.3628]{Samko1}:
\begin{equation} \label{Example-2-M}
M(t) = h_{\alpha,\lambda}(t) \, = \, 
\frac{t^{\alpha-1}}{\Gamma(\alpha)} e^{- \lambda \, t}
\end{equation}
\begin{equation} \label{Example-2-K}
K(t) = h_{1-\alpha,\lambda}(t) \, + \, 
\frac{\lambda^\alpha}{\Gamma(1-\alpha)} \, 
\gamma(1-\alpha,\lambda t),
\end{equation}
and vice versa,
where $0<\alpha <1$, and $\lambda \ge 0$, $t>0$,
and $\gamma(\beta,t)$ is the incomplete gamma function
\begin{equation} 
\gamma(\beta,t) = \int^t_0 \tau^{\beta -1} e^{-\tau}\, d\tau , 
\end{equation}
where $t>0$.
\end{Example}

%%% Example 3.
\begin{Example} 
The following pair of kernels
\begin{equation} \label{Example-3-M}
M(t) = (\sqrt{t})^{\alpha-1} \, J_{\alpha-1}(2\sqrt{t}) , 
\end{equation}
\begin{equation} \label{Example-3-K}
K(t) = (\sqrt{t})^{-\alpha} \, I_{-\alpha}(2\sqrt{t}),
\end{equation} 
and vice versa, are Sonin kernels 
(see \cite{Sonine-1,Sonine-2},\cite[p.3627]{Samko1})
if $0<\alpha <1$, where
\begin{equation} 
J_\nu (t) = \sum^{\infty}_{k=0} 
\frac{(-1)^k(t/2)^{2k+\nu}}{k!\Gamma(k+\nu+1)},
\quad
I_\nu (t) = \sum_{k=0}^{\infty} 
\frac{(t/2)^{2k+\nu}}{k!\Gamma(k+\nu+1)}
\end{equation}
are the Bessel and the modified Bessel functions, respectively. 
\end{Example}

%%% Example 4.
\begin{Example} 
As special case of kernels 
\eqref{Example-3-M} and \eqref{Example-3-K}, we can consider the pair
\begin{equation} \label{Example-4-M}
M(t) = \frac{\cos(2 \sqrt{t})}{\sqrt{\pi \, t}} , 
\end{equation}
\begin{equation}
K(t) = \frac{\cosh(2 \sqrt{t})}{\sqrt{\pi \, t}} ,
\end{equation}
and vice versa. 
\end{Example}

%%% Example 5.
\begin{Example} 
The example of the Sonin pair is the kernel pair \cite{Hanyga} 
\begin{equation} \label{Example-5-N}
M(t) = h_{1-\beta+\alpha}(t) \, + \, h_{1-\beta}(t),
\end{equation}
\begin{equation} \label{Example-5-K}
K(t) = t^{\beta -1} \, E_{\alpha,\beta} [-t^{\alpha} ],
\end{equation}
where $0<\alpha < \beta <1$, and $E_{\alpha,\beta}[z]$ is the two-parameters Mittag-Leffler function 
\begin{equation} \label{ML}
E_{\alpha,\beta}[z] = \sum^{\infty}_{k=0} 
\frac{z^k}{\Gamma(\alpha\, k + \beta)}, 
\end{equation}
where $\alpha >0$, and $\beta, z\in \mathbb{C}$.
\end{Example}

%%% Example 6.
\begin{Example} 
The following Sonin pair is described by the kernels 
(see Eq. 7.15 in \cite[p.3629]{Samko1}):
\begin{equation} \label{Example-6-M}
M(t) = t^{\alpha-1} \Phi(\beta, \alpha; -\lambda \, t) , 
\end{equation}
\begin{equation} \label{Example-6-K}
K(t) = \frac{\sin (\pi \alpha)}{\pi} \, t^{-\alpha} \, 
\Phi(-\beta,1- \alpha; -\lambda \, t) , 
\end{equation}
or vice versa, where $0<\alpha <1$, and $\Phi(\beta, \alpha; z)$ is the Kummer function
\begin{equation} 
\Phi(\beta, \alpha; z) = \sum^{\infty}_{k=0} 
\frac{(\beta)_k}{(\alpha)_k} \, \frac{z^k}{k!} .
\end{equation}
\end{Example}

%%% Example 7.
\begin{Example} 
The following pair of the Sonin kernels 
(see Eq. 7.16, 7.18 in \cite[p.3629]{Samko1}):
\begin{equation} \label{Example-7-M}
M(t) = 1 \, + \, \frac{\lambda}{\Gamma(\alpha) \sqrt{t}} , 
\end{equation}
\begin{equation} \label{Example-7-K}
K(t) = \frac{1}{ \sqrt{\pi t}} \, - \, 
\lambda \, e^{\lambda^2 \, t} \, \text{erfc} (\lambda \, \sqrt(t)) , 
\end{equation}
or vice versa, where $\lambda >0$, and $\text{erfc}(z)$ is the complementary error function
\begin{equation}
\text{erfc} (z) = 1 \, - \, \text{erf} (z) = 1- \frac{2}{\sqrt{\pi}} \int^{t}_0 e^{-z^2} \, dz . 
\end{equation}
\end{Example}

%%% Example 8.
\begin{Example} 
The following Sonin pair of the kernels 
(see Eq. 7.17, 7.19 in \cite[pp.3629-3630]{Samko1}):
\begin{equation} \label{Example-8-M}
M(t) = 1 \, - \, \frac{\lambda}{\Gamma(\alpha)} \, t^{\alpha-1} , 
\end{equation}
\begin{equation} \label{Example-8-K}
K(t) = \lambda \, t^{-\alpha} \, 
E_{1-\alpha,1-\alpha} [\lambda \, t^{1-\alpha}] ,
\end{equation}
or vice versa, where $\lambda>0$.
\end{Example}

%%% Example 9.
\begin{Example} 
As a Sonin's kernels, we can use some functions with the power-logarithmic singularities at the origin \cite[pp.3627-3630]{Samko1}. 
The following pair of the Sonin kernels 
(see Eq. 7.22-7.24 in \cite[p.3630]{Samko1}) with the power-logarithmic singularities:
\begin{equation} \label{Example-9-M}
M(t) = \frac{A-\ln(t)}{\Gamma(\alpha)} t^{\alpha-1} , 
\end{equation}
\begin{equation} \label{Example-9-K}
K(t) = \mu_{\alpha,h} (t) = \int^{\infty}_0 
\frac{t^{z-\alpha} \, e^{h\, z}}{\Gamma(z+1-\alpha)} dz ,
\end{equation}
or vice versa, where $\mu_{\alpha,h}(t)$ is the Volterra function with
\[
h=\frac{\Gamma^{\prime}(\alpha)}{\Gamma(\alpha)} \, - \, A .
\]
We assume that nonlocal discrete maps, which are derived from equations with GFD and GFI for these kernels, are described by the same equations as for kernels belong to $C_{-1}(0,\infty)$.
\end{Example}

The discrete mappings with an arbitrary function $N={\cal G}(t,X(t))$ are called universal in nonlinear dynamics. 
We can also called the proposed mappings of the GFDynamics
as universal general maps with nonlocality in time.
The universality of the mapping is due to the use of an arbitrary nonlinear function $N={\cal G}(t,X(t))$, and the generality is due to the use of almost arbitrary Sonin kernels of operators $M(t)$ and $K(t)$.

%%%%%%%%%%%%%%%%%%%%%%%%%%%%%%%%%%%%%%%%%%%%%%%%%%%%%%%%%%
%%%%%%%%%%%%%%%%%%%%%%%%%%%%%%%%%%%%%%%%%%%%%%%%%%%%%%%%%%
%%%%%%%%%%%%%%%%%%%%%%%%%%%%%%%%%%%%%%%%%%%%%%%%%%%%%%%%%%

\section{General Fractional Dynamics of Arbitrary Orders}

In this section, the general fractional dynamics will be described by equations with GFI and GFD of arbitrary order, proposed in the work of Yu. Luchko \cite{Luchko2021-2}.

\subsection{General Fractional Calculus of Arbitrary Orders}

In paper \cite{Luchko2021-2}, the Sonin set of the kernel pairs is defined in the following form. 

\begin{Definition} \label{DEF-Sonine}
Let the functions $\mu(t),\, \nu(t)$ belong to the function space
\begin{equation} \label{C(-10)}
C_{-1,0}(0,\infty)\, := \, \{X:\ X(t) = t^{p} Y(t), \ 
t \in (0, \infty),\ -1<p<0, \ Y(t) \in C[0,\infty)\} ,
\end{equation} 
and the Sonin condition 
\begin{equation} \label{SonineCondition}
( \mu \, *\, \nu )(t) = \{1 \} 
\end{equation}
holds for $t \in (0, \infty)$. 
The set of such kernel pairs is called the Sonin set $\mathcal{S}_{-1}$.
For the kernel $\mu (t)$, the kernel $\nu(t)$ is called its associate Sonin kernel.
\end{Definition}

In paper \cite{Luchko2021-2}, a generalization of Definition \ref{DEF-Sonine} was proposed. 

\begin{Definition} \label{DEF-Luchko}
Let the function $M(t), \, K(t)$ belongs to the function spaces
\begin{equation}
M(t) \in C_{-1}(0,\infty), \quad 
K(t) \in C_{-1,0}(0,\infty),
\end{equation}
where $C_{-1,0}(0,\infty)$ is defined by \eqref{C(-10)}, 
\begin{equation} 
C_{-1}(0,\infty)\, := \, \{X:\ X(t) = t^{p} Y(t), \ 
t \in (0, \infty),\ p>-1, \ Y(t) \in C[0,\infty)\} ,
\end{equation} 
and the Luchko condition 
\begin{equation} \label{LuchkoCondition}
( M \, * \, K )(t) = \{1 \}^n = 
h_{n}(t) = \frac{t^{n-1}}{(n-1)!} 
\end{equation}
holds for $t \in (0, \infty)$. 
The set of such kernel pairs is called the Luchko set $\mathcal{L}_{n,0}$.
\end{Definition}

In paper \cite{Luchko2021-2}, it was proposed an approach for construction a pair $(M(t),\, K(t))$ of the kernels from the Luchko set $\mathcal{L}_{n,0}$ by using the kernels 
$\mu (t),\, \nu(t)$ from the Sonin set $\mathcal{S}_{-1}$. 
This approach is formulated by the following theorem that was proved in \cite{Luchko2021-2}.

\begin{Theorem} \label{Theorem-MK-1}
Let $(\mu(t),\ \nu (t))$ be a Sonin pair of kernels from $\mathcal{S}_{-1}$. 
Then the pair $(M(t),\ K(t))$ of the kernels
\begin{equation} \label{Luchko-GFC2}
M(t) \, = \, (\{1\}^{n-1}\, *\, \mu)(t), \quad 
K (t) \, = \, \nu(t)
\end{equation}
belongs to the set $\mathcal{L}_{n,0}$.
\end{Theorem}

Theorem \ref{Theorem-MK-1} is proved by Yu. Luchko in \cite{Luchko2021-2}.

Using Definition \ref{DEF-Luchko} and Theorem \ref{Theorem-MK-1},
the GFI and GFD of arbitrary order are defined in \cite{Luchko2021-2}. 

\begin{Definition} \label{GFOperators}
Let $(M(t),\ K(t))$ be pair of the kernels \eqref{Luchko-GFC2} from the Luchko set $\mathcal{L}_{n,0}$. 
The GFI with the kernel $M(t)$ is 
\begin{equation} \label{L-def-GFI}
I^{t,n}_{(M)} [\tau]\, X(\tau) \, := 
\int^t_0 d \tau \, M(t-\tau) \, X(\tau) \, = \, 
\int^t_0 d \tau \, (\{1\}^{n-1} \, *\, \mu)(t-\tau) \, X(\tau)
, \end{equation}
where $t >0$.
The GFDs with the kernel $K(t)$ are 
\begin{equation} \label{L-def-GFD-1} 
D^{t,n}_{(K)}[\tau] X(\tau) := 
\frac{d^n}{dt^n}\, \int^t_0 d\tau \, K(t-\tau) \, X(\tau) 
\, = \, 
\frac{d^n}{dt^n}\, \int^t_0 d\tau \, 
\nu(t-\tau) \, X(\tau) ,
\end{equation}
\begin{equation} \label{L-def-GFD-2} 
D^{t,n,*}_{(K)}[\tau] X(\tau) := 
\ \int^t_0 d\tau \, K(t-\tau) \, X^{(n)} (\tau) \, = \, 
\int^t_0 d\tau \, \nu(t-\tau) \, X^{(n)}(\tau) ,
\end{equation}
where $t >0$ and $ X^{(n)} (\tau) = d^n X(\tau)/ d\tau^n$.
\end{Definition}

%%%%%%%%%%%%%%%%%%%%%%%%%%%%%%%%%%%%%%%%%%%%%%%%%%%%%%%%%%%

In paper \cite{Luchko2021-2}, the fundamental theorems of GFC for the GFI \eqref{L-def-GFI} and GFDs
\eqref{L-def-GFD-1}, \eqref{L-def-GFD-2} are proved. 

\begin{Theorem}[First Fundamental Theorem for GFC of arbitrary order] \label{Theorem-7}

Let $(M(t),\ K(t))$ be a kernel pair from the Luchko set $\mathcal{L}_{n,0}$. 

If $X(t) \in C_{-1}(0,\infty)$, then
\begin{equation} \label{FTLn}
D^{t,n}_{(K)} [\tau] \, I^{\tau,n}_{(M)} [s] \, X(s) = X(t) .
\end{equation}

If $X(t) \in C_{-1,(K)}(0,\infty)$, then 
\begin{equation} \label{FTCn}
D^{t,n,*}_{(K)} [\tau] \, I^{\tau,n}_{(M)} [s] \, X(s) = X(t) ,
\end{equation}
where function $X(t)$ belongs to $C_{-1,(K)}(0,\infty)$, if it can be represented as $X(t) = I^{t,n}_{(K)} [\tau] \, Y(\tau)$, where $Y(t) \in C_{-1}(0,\infty)$. 
\end{Theorem}

\begin{Theorem}[Second Fundamental Theorem for GFC of arbitrary order] \label{Theorem-8}

Let $(M(t),\ K(t))$ be a kernel pair from the Luchko set $\mathcal{L}_{n,0}$. 

If $X(t) \in C_{-1}^n(0,\infty)$, then
\begin{equation} \label{2FTCn}
I^{t,n}_{(M)}[\tau] \, D^{\tau,n,*}_{(K)}[s] \, X(s) = 
X(t) - \sum^{n-1}_{k=0} X^{(k)}(0)\, h_{k+1}(t) .
\end{equation}

If $X(t) \in C^n_{-1,(M)} (0,\infty)$, then
\begin{equation} \label{2FTLn}
I^{t,n}_{(M)}[\tau] \, D^{\tau,n}_{(K)}[s] \, X(s) = X(t) ,
\end{equation}
where function $X(t)$ belongs to $C^n_{-1,M}(0,\infty)$, if it can be represented in the form $X(t) = I^{t,n}_{(M)} [\tau] \, Y(\tau) \in C^{-1}(0,\infty)$, where $Y(t) \in C_{-1}(0,\infty)$.
\end{Theorem}

Theorems \ref{Theorem-7} and \ref{Theorem-8} are proved by Yu. Luchko in \cite{Luchko2021-2}.

%%%%%%%%%%%%%%%%%%%%%%%%%%%%%%%%%%%%%%%%%%%%%%%%%%%%%%%%%%

\subsection{General Momenta of Arbitrary Orders}

In fractional dynamics, to describe nonlocal mappings of arbitrary order, we should define generalized momenta \cite{Springer2010}.
For generalized fractional dynamics, we also should use the generalized momenta. We will use the following definitions.

\begin{Definition} \label{GFMomenta}
Let $(M(t),\ K(t))$ be a pair of the kernels from the Luchko set $\mathcal{L}_{n,0}$ in the form 
\begin{equation} \label{Luchko-GFC2-new}
M(t) \, = \, (\{1\}^{n-1}\, *\, \mu)(t), \quad 
K (t) \, = \, \nu(t) .
\end{equation}
Then the generalized momentum $V_k(t)$ is defined as
\begin{equation} \label{P-GFI}
V_k(t) \, := \, 
I^{t,k}_{(h)} [\tau]\, X(\tau) \, = \,
\int^t_0 d\tau \, h_{k}(t-\tau) \, X(\tau), 
\end{equation}
where $X(t) \in C_{-1}(0,\infty)$.
The generalized momenta $P_k(t)$, $P^{*}_k(t)$ are defined by the equations
\begin{equation} \label{P-GFD-1} 
P_k(t) \, := \,
D^{t,k}_{(K)}[\tau] X(\tau) \, = \, 
\frac{d^{k}}{dt^{k}} \, \int^t_0 d\tau \, 
K(t-\tau) \, X(\tau) ,
\end{equation}
\begin{equation} \label{P-GFD-2} 
P^{*}_k(t) \, := \, \frac{d^k}{dt^k} X(t) ,
\end{equation}
where $k=1, \dots ,n-1$, $t>0$ and
\begin{equation} 
h_{n}(t) = \frac{t^{n-1}}{(n-1)!} .
\end{equation}
To simplify, we can denine
\begin{equation} \label{P-GFD-0} 
V_0(t) \, := \, X(t) , \quad
P_0(t) \, := \, X(t) , \quad
P^{*}_0(t) \, := \, X(t) , 
\end{equation}
where $X(t) \in C^{k}_{-1}(0,\infty)$.
\end{Definition}

Let us prove the properties of the generalized momenta that describes connections of GFI and GFDs of the variable $X(t)$ with the momenta $V_k(t)$, $P_k(t)$, and $P^{*}_k(t)$.

\begin{Theorem} \label{P-Theorem}

Let $V_k(t)$ be defined by \eqref{P-GFI}. 
Then the equation
\begin{equation} \label{V-GFI-Theorem}
I^{t,n}_{(M)} [\tau]\, X(\tau) \, = \,
I^{t,n-k}_{(M)} [\tau]\, V_k(\tau) ,
\end{equation}
holds, if $X(t) \in C_{-1}(0,\infty)$,
and $V_k(t)\in C_{-1}(0,\infty)$.

Let $P_k(t)$, and $P^{*}_k(t)$ be defined by 
\eqref{P-GFD-1} and \eqref{P-GFD-2}, respectively. 
Then equalities
\begin{equation} \label{P-GFD-1-Theorem} 
D^{t,n}_{(K)}[\tau] X(\tau) \, = \, 
\frac{d^{n-k}}{dt^{n-k}} P_k(t) ,
\end{equation}
\begin{equation} \label{P-GFD-2-Theorem} 
D^{t,n,*}_{(K)}[\tau] X(\tau) \, = \, 
D^{t,n-k,*}_{(K)}[\tau] P^{*}_k(\tau) ,
\end{equation}
are satisfied, if $X(t) \in C_{-1}(0,\infty)$, 
and $P_k(t), \, P^{*}_k(t) \in C^{k}_{-1}(0,\infty)$.
\end{Theorem}

\begin{proof} 

(A) Let us first prove equation \eqref{V-GFI-Theorem} for the GFI.
We can use the fact that the standard integration of order $n \in \mathbb{N}$ ($n$-fold integral) has the form \cite[p.33]{FC1}
\begin{equation}
\int^t_0 d\tau_1 \int^{\tau_1}_0 d\tau_2 \dots 
\int^{\tau_{n-1}}_0 d\tau_n \, X (\tau_n) =
\int^{t}_0 d\tau \, h_{n} (t-\tau) \, X (\tau) = 
I^{t,n}_{(h)} [\tau] X(\tau) ,
\end{equation}
where
\begin{equation}
h_{n} (t-\tau) \, = \, \frac{(t-\tau)^{n-1}}{(n-1)!} .
\end{equation}
Therefore, the GFI that is given by \eqref{L-def-GFI} with $n>1$ is 
\[
I^{t,n}_{(M)} [\tau]\, X(\tau) \, = \,
\int^t_0 d \tau \, ( h_{n-1} \, * \, \mu)(t-\tau) \, X(\tau)
 \, = \, 
\]
\[
(( h_{n-1} \, * \, \mu )\, * \, X)(t) \, = \, 
( h_{n-1} \, * \, (\mu \, * \, X))(t) \, = \, 
\]
\begin{equation} \label{L-def-GFI-new}
\int^{t}_0 d\tau \, h_{n-1} (t-\tau) \, (\mu \, * \, X)(\tau)
\, = \, 
I^{t,n}_{(h)} [\tau] \, (\mu \, * \, X)(\tau) .
\end{equation}
Using that 
\begin{equation}
I^{t,n}_{(h)} [\tau] X(\tau)\, = \,
I^{t,n-k}_{(h)} [\tau] \, I^{\tau,k}_{h} [s] X(s) ,
\end{equation}
we get
\begin{equation}
I^{t,n-k}_{(M)} [\tau] \, V_k (\tau) \, = \, 
I^{t,n-k}_{(M)} [\tau] \, 
I^{\tau,k}_{(h)} [s]\, X(s) \, = \,
I^{t,n}_{(M)} [\tau]\, X(\tau) .
\end{equation}

(B) Let us prove equation \eqref{P-GFD-1-Theorem} by using definitions \eqref{P-GFD-1} and \eqref{L-def-GFD-1}. 
Equation \eqref{P-GFD-1-Theorem} is proved by the following transformations 
\[ 
\frac{d^{n-k}}{dt^{n-k}} P_k(t) \, = \,
\frac{d^{n-k}}{dt^{n-k}}
D^{t,k}_{(K)}[\tau] X(\tau) \, = \,
\frac{d^{n-k}}{dt^{n-k}} 
\frac{d^{k}}{dt^{k}} \, \int^t_0 d\tau \, 
K(t-\tau) \, X(\tau) \, =
\]
\begin{equation} \label{Proof-RL-1} 
\frac{d^{n}}{dt^{n}} \, \int^t_0 d\tau \, 
K(t-\tau) \, X(\tau) \, = \, 
D^{t,n}_{(K)}[\tau] X(\tau) . 
\end{equation}

Let us prove equation \eqref{P-GFD-2-Theorem} by using definitions \eqref{P-GFD-2} and \eqref{L-def-GFD-2}. 
Equation \eqref{P-GFD-2-Theorem} is proved by the transformations
\[
D^{t,n-k,*}_{(K)}[\tau] P^{*}_k(t) \, = \,
\int^t_0 d\tau \, K(t-\tau) \, (P^{*}_k)^{(n-k)} (\tau) \, =
\]
\begin{equation} \label{Proof-RL-2} 
\int^t_0 d\tau \, K(t-\tau) \, (X^{(k)})^{(n-k)} (\tau) \, = \,
\int^t_0 d\tau \, K(t-\tau) \, X^{(n)} (\tau) \, =
D^{t,n}_{(K)}[\tau] X(\tau) . 
\end{equation}

\end{proof}

%%%%%%%%%%%%%%%%%%%%%%%%%%%%%%%%%%%%%%%%%%%%%%%%%%%%%%%%%%
%%%%%%%%%%%%%%%%%%%%%%%%%%%%%%%%%%%%%%%%%%%%%%%%%%%%%%%%%%
%%%%%%%%%%%%%%%%%%%%%%%%%%%%%%%%%%%%%%%%%%%%%%%%%%%%%%%%%%

Theorem \ref{P-Theorem} allows us to derive exact solutions of equations with GFI and GFDs of arbitrary orders and periodic kicks for $V_k(t)$, $P_k(t)$, $P^{*}_{k}(t)$ with $k=1, \dots , n-1$ and the variable $X(t)$.

\begin{Theorem} \label{Theorem-10}
Let the functions $K(t)$, $M(t)$ be a Luchko pair of kernels from ${\cal L}_{n,0}$, and $X(t) \in C_{-1}(0,\infty)$.
Then the equations
\be \label{EQ-Th10-1a}
I^{t,n}_{(M)} [\tau] \, X(\tau) = \lambda \, {\cal G}(t,X(t)) 
\sum^{\infty}_{j=1} \delta ((t+\epsilon)/T-j) ,
\ee
\be \label{EQ-Th10-1b}
I^{t,n-k}_{(M)} [\tau] \, V_k(\tau) = \lambda \, {\cal G}(t,X(t)) 
\sum^{\infty}_{j=1} \delta ((t+\epsilon)/T-j) ,
\ee
where $k=1, \dots , n-1$, have the solutions
\be \label{EQ-Th10-2a}
X(t) \, = \, \lambda \, T \, \sum^{m}_{j=1} 
K^{(n)}(t-Tj+\epsilon)) \, {\cal G}(Tj-\epsilon,X(Tj-\epsilon)) ,
\ee
\be \label{EQ-Th10-2b}
V_k(t) \, = \, \lambda \, T \, \sum^{m}_{j=1} 
K^{(n-k)}(t-Tj+\epsilon)) \, {\cal G}(Tj-\epsilon,X(Tj-\epsilon)) ,
\ee
if $Tm<t<T(m+1)$.
For the time points $s = Tj-\epsilon$ at $\epsilon \to 0+$, and the variables 
\begin{equation} \label{EQ-Th10-XVk}
X_{j}=\lim_{\epsilon \rightarrow 0+} X(Tj-\epsilon), \quad
V_{k,j}=\lim_{\epsilon \rightarrow 0+} V_k(Tj-\epsilon),
\quad (j=1,...,m+1) ,
\end{equation} 
solutions \eqref{EQ-Th10-2a}, \eqref{EQ-Th10-2b} are represented by the general nonlocal mappings
\be \label{EQ-Th10-3a}
X_{m+1} = X_m \,+ \, 
\lambda \, T \, K^{(n)}(T) \, {\cal G}(Tm,X_m) +
\lambda \, T \, \sum^{m-1}_{j=1} 
\Omega^{(n)}_{(K)} (T,m-j) \, {\cal G}(jT,X_j) ,
\ee
\be \label{EQ-Th10-3b}
V_{k,m+1} = V_{k,m} \,+ \, 
\lambda \, T \, K^{(n-k)}(T) \, {\cal G}(Tm,X_m) +
\lambda \, T \, \sum^{m-1}_{j=1} 
\Omega^{(n-k)}_{(K)} (T,m-j) \, {\cal G}(jT,X_j) ,
\ee
where $K(t) \in C_{-1,0}(0,\infty)$ is a function that is associated kernel to the kernel $M(t)$, and
\be \label{Omega-10}
\Omega_{(K)} (T,z) \, := \,
K(T(z+1) ) \, - \, K( Tz ) , \quad
\Omega^{(k)}_{(K)} (T,z) \, := \, 
\frac{1}{T^k} \frac{d^k}{dz^k} \Omega_{(K)} (T,z) .
\ee
 \end{Theorem}

\begin{proof} 
In this proof, we use the first fundamental theorem for GFC of arbitrary order that is proved in \cite{Luchko2021-2}, which 
states that if $(M(t),\ K(t))$ is a kernel pair from the Luchko set $\mathcal{L}_{n,0}$ and $X(t) \in C_{-1}(0,\infty)$, then
\begin{equation} \label{Second-FT-AO}
D^{t,n}_{(K)} [\tau] \, I^{\tau,n}_{(M)} [s] \, X(s) = X(t) 
\end{equation}
holds for $t>0$.

%%% Part 1
Applying the GFD $D^{s,n-k}_{(K)} [t]$ with the kernal $K(t)$, which is associated to the kernel $M(t)$ of the GFI $I^{t,n}_{(M)}$,
to equation \eqref{EQ-Th10-1a} and \eqref{EQ-Th10-1b}, we derive
\be \label{Proof-Th10-1}
D^{s,n-k}_{(K)} [t]
I^{t,n-k}_{(M)} [\tau] \, V_k(\tau) = \lambda \, 
D^{s,n-k}_{(K)} [t] \, {\cal G}(t,X(t)) 
\sum^{\infty}_{j=1} \delta ((t+\epsilon)/T-j) ,
\ee
where $k=1, \dots , n-1$, and $V_k(\tau):=X(\tau)$ for $k=0$.
Using \eqref{Second-FT-AO}, equation \eqref{Proof-Th10-1} is written as
\be \label{Proof-Th10-2}
V_k(s) = \lambda \, 
D^{s,n-k}_{(K)} [t] \, {\cal G}(t,X(t)) 
\sum^{\infty}_{j=1} \delta ((t+\epsilon)/T-j) ,
\ee
Using definition of $D^{s,n-k}_{(K)} [t]$, we get
\be \label{Proof-Th10-3}
V_k(s) = \lambda \, 
\frac{d^{n-k}}{ds^{n-k}} \int^s_0 dt \, K(s-t) \, 
 {\cal G}(t,X(t)) \sum^{\infty}_{j=1} \delta ((t+\epsilon)/T-j) ,
\ee
For $Tm<s<T(m+1)$, equation \eqref{Proof-Th10-3} is 
\be \label{Proof-Th10-5a}
V_k(s) = \lambda \, \sum^{m}_{j=1}
\frac{d^{n-k}}{ds^{n-k}} \int^s_0 dt \, K(s-t) \, 
 {\cal G}(t,X(t)) \delta ((t+\epsilon)/T-j) ,
\ee
Using equality \eqref{Proof-delta}, equation \eqref{Proof-Th10-5a} gives
\be \label{Proof-Th10-5b}
V_k(s) = \lambda \, T \, \sum^{m}_{j=1}
\frac{d^{n-k}}{ds^{n-k}} \, K(s-(Tj-\epsilon)) \, 
 {\cal G}(Tj-\epsilon,X(Tj-\epsilon)) ,
\ee
that can be writen as
\be \label{Proof-Th10-5c}
V_k(s) = \lambda \, T \, \sum^{m}_{j=1}
K^{(n-k)}(s-(Tj-\epsilon)) \, 
 {\cal G}(Tj-\epsilon,X(Tj-\epsilon)) ,
\ee
where $K^{(k)}(z)= d^{(k)} K(z)/dz^{k}$. 
Equation \eqref{Proof-Th10-5c} gives solutions
\eqref{EQ-Th10-2a} and \eqref{EQ-Th10-2b}.

%%% Part 2
Using \eqref{Proof-Th10-5c} for $s= T(m+1)-\epsilon$ and
$s= Tm-\epsilon$, we obtain
\be \label{Proof-Th10-6a}
V_k(T(m+1)-\epsilon) = \lambda \, T \, \sum^{m}_{j=1}
K^{(n-k)}(T(m+1)-Tj) \, {\cal G}(Tj-\epsilon,X(Tj-\epsilon)) ,
\ee
\be \label{Proof-Th10-6b}
V_k(Tm-\epsilon) = \lambda \, T \, \sum^{m-1}_{j=1}
K^{(n-k)}(Tm-Tj) \, {\cal G}(Tj-\epsilon,X(Tj-\epsilon)) ,
\ee
Using variables \eqref{EQ-Th10-XVk}, solutions \eqref{Proof-Th10-6a} and \eqref{Proof-Th10-6b} at the limit $\epsilon \to 0+$ give
\be \label{Proof-Th10-7a}
V_{k,m+1} = \lambda \, T \, \sum^{m}_{j=1}
K^{(n-k)}(T(m+1-j)) \, {\cal G}(Tj,X_j) ,
\ee
\be \label{Proof-Th10-7b}
V_{k,m} = \lambda \, T \, \sum^{m-1}_{j=1}
K^{(n-k)}(T(m-j)) \, {\cal G}(Tj,X_j) .
\ee
Subtracting equation \eqref{Proof-Th10-7b}  from equation \eqref{Proof-Th10-7a}, we obtain 
\[
V_{k,m+1} - V_{k,m} = 
\lambda \, T \, K^{(n-k)}(T) \, {\cal G}(Tm,X_m) \, + \, 
\]
\be \label{Proof-Th10-8}
\lambda \, T \, \sum^{m-1}_{j=1}
\Bigl( K^{(n-k)}(T(m+1-j)) - K^{(n-k)}(T(m-j)) \Bigr)
\, {\cal G}(Tj,X_j) .
\ee
where $k=0,1,\dots , n-1$.
Then using \eqref{Omega-10}, equation \eqref{Proof-Th10-8}, takes form \eqref{EQ-Th10-3a}, \eqref{EQ-Th10-3b}.

\end{proof}

%%%%%%%%%%%%%%%%%%%%%%%%%%%%%%%%%%%%%%%%%%%%%%%%%%%%%%%%%%
%%%%%%%%%%%%%%%%%%%%%%%%%%%%%%%%%%%%%%%%%%%%%%%%%%%%%%%%%%
%%%%%%%%%%%%%%%%%%%%%%%%%%%%%%%%%%%%%%%%%%%%%%%%%%%%%%%%%%

Let us derive general nonlocal maps from the equations with the GFD $D^{t,n}_{(K)}$ of arbitrary orders and periodic kicks.

\begin{Theorem} \label{Theorem-11}
Let $K(t)$, $M(t)$ be a pair of kernels from the Luchko set ${\cal L}_{n,0}$, and $X(t) \in C^m_{-1}(0,\infty)$.
Then the equations
\be \label{EQ-Th11-1a}
D^{t,n}_{(K)} [\tau] \, X(\tau) = \lambda \, {\cal G}(t,X(t)) 
\sum^{\infty}_{j=1} \delta ((t+\epsilon)/T-j) ,
\ee
\be \label{EQ-Th11-1b}
\frac{d^{n-k}}{dt^{n-k}} \, P_k(\tau) = \lambda \, {\cal G}(t,X(t)) 
\sum^{\infty}_{j=1} \delta ((t+\epsilon)/T-j) ,
\ee
where $k=1, \dots , n-1$, have the solutions
\be \label{EQ-Th11-2a}
X(t) \, = \, \lambda \, T \, \sum^{m}_{j=1} 
M(t-Tj+\epsilon) \, {\cal G}(Tj-\epsilon,X(Tj-\epsilon)) ,
\ee
\be \label{EQ-Th11-2b}
P_k(t) \, = 
\sum^{n-k-1}_{j=0} P^{(j)}_k(0) \, h_{j+1} (t) \, + \, 
\lambda \, T \, \sum^{m}_{j=1} 
h_{n-k}(t-Tj+\epsilon) \, {\cal G}(Tj-\epsilon,X(Tj-\epsilon)) ,
\ee
if $Tm<t<T(m+1)$.
For the time points $s = Tj-\epsilon$ at $\epsilon \to 0+$, and 
the variables 
\begin{equation} \label{EQ-Th11-XPk}
X_{j}=\lim_{\epsilon \rightarrow 0+} X(Tj-\epsilon), \quad
P_{k,j}=\lim_{\epsilon \rightarrow 0+} P_k(Tj-\epsilon),
\quad (j=1,...,m+1) ,
\end{equation} 
solutions of equations \eqref{EQ-Th11-1a}, \eqref{EQ-Th11-1b} are represented by the general nonlocal mappings
\be \label{EQ-Th11-3a}
X_{m+1} = X_m \,+ \, \lambda \, T \, M(T) \, {\cal G}(Tm,X_m) +
\lambda \, T \, \sum^{m-1}_{j=1} 
\Omega_{(M)} (T,m-j) \, {\cal G}(jT,X_j) ,
\ee
\[
P_{k,m+1} = P_{k,m} \, + \, \sum^{n-k-1}_{j=0} P^{(j)}_k(0) \, 
\Omega_{(h_{j+1})}(T,m) \, + \, 
\lambda \, T \, h_{n-k}(T) \, {\cal G}(Tm,X_m) \, +
\]
\be \label{EQ-Th1-3b}
\lambda \, T \, \sum^{m-1}_{j=1} 
\Omega_{(h_{n-k})} (T,m-j) \, {\cal G}(jT,X_j) ,
\ee
where $M(t) \in C_{-1}(0,\infty)$ is a function that is associated kernel to the kernel $K(t)$,
the function $\Omega_{(M)}(T,z)$ is defined by \eqref{Omega-10},
and $h_{\alpha}(t) = t^{\alpha-1}/\Gamma(\alpha)$ 
with $\alpha>0$. 
\end{Theorem}

\begin{proof} In the proof, we use the second fundamental theorem for GFC of arbitrary order that is proved in \cite{Luchko2021-2}.
This theorem states that if $(M(t),\ K(t))$ is a kernel pair from the Luchko set $\mathcal{L}_{n,0}$ and 
$X(t) \in C^n_{-1,(M)} (0,\infty)$, then
\begin{equation} \label{2FTLn-b}
I^{t,n}_{(M)}[\tau] \, D^{\tau,n}_{(K)}[s] \, X(s) = X(t) .
\end{equation}
The function $X(t)$ belongs to $C^n_{-1,(M)}(0,\infty)$, if it can be represented in the form $X(t) = I^{t,n}_{(M)} [\tau] \, Y(\tau) \in C_{-1}(0,\infty)$, where $Y(t) \in C_{-1}(0,\infty)$.

The action of the integral operator $I^{s,n}_{(M)}[t]$ on equation \eqref{EQ-Th11-1a}, and 
the operator $I^{t,n-k}_{(h)}[t]$ on \eqref{EQ-Th11-1b} give
\be \label{Proof-Th11-1a}
I^{s,n}_{(M)}[t] \, D^{t,n}_{(K)} [\tau] \, X(\tau) = 
\lambda \, I^{s,n}_{(M)}[t] \, {\cal G}(t,X(t)) 
\sum^{\infty}_{j=1} \delta ((t+\epsilon)/T-j) ,
\ee
\be \label{Proof-Th11-1b}
I^{s,n-k}_{(h)}[t]
\frac{d^{n-k}}{dt^{n-k}} \, P_k(\tau) = \lambda \, 
I^{s,n-k}_{(h)}[t] \, {\cal G}(t,X(t)) 
\sum^{\infty}_{j=1} \delta ((t+\epsilon)/T-j) ,
\ee
where $k=1, \dots , n-1$, and
\[
I^{s,n-k}_{(h)}[t] X(t) \, = \, 
\int^s_0 dt \, h_{n-k}(s-t) \, X(t) .
\]

%%%%%%%%%%%%%%%%%%%%

Using the second fundamental theorems for $D^{t,n}_{(K)}$ 
in form \eqref{2FTLn-b}, and ${d^{n-k}}/{dt^{n-k}}$, 
equations \eqref{Proof-Th11-1a} and \eqref{Proof-Th11-1b} give
\be \label{Proof-Th11-2a}
X(s) \, = \,
\lambda \, I^{s,n}_{(M)}[t] \, {\cal G}(t,X(t)) 
\sum^{\infty}_{j=1} \delta ((t+\epsilon)/T-j) ,
\ee
\be \label{Proof-Th11-2b}
P_k(s) \, - \, \sum^{n-k-1}_{j=0} P^{(j)}_k(0) \, h_{j+1} (s) \, = \, \lambda \, 
I^{s,n-k}_{(h)}[t] \, {\cal G}(t,X(t)) 
\sum^{\infty}_{j=1} \delta ((t+\epsilon)/T-j) .
\ee
Using definition of $I^{s,n-k}_{(M)}$ and $I^{s,n-k}_{(h)}$, 
equations \eqref{Proof-Th11-2a} and \eqref{Proof-Th11-2b} 
 are written as
\be \label{Proof-Th11-3a}
X(s) \, = \,
\lambda \, \int^s_0 dt \, M_{n}(s-t) \, {\cal G}(t,X(t)) 
\sum^{\infty}_{j=1} \delta ((t+\epsilon)/T-j) ,
\ee
\be \label{Proof-Th11-3b}
P_k(s) \, - \, \sum^{n-k-1}_{j=0} P^{(j)}_k(0) \, h_{j+1} (s) \, = \, \lambda \, 
\int^s_0 dt \, h_{n-k}(s-t) \, {\cal G}(t,X(t)) 
\sum^{\infty}_{j=1} \delta ((t+\epsilon)/T-j) ,
\ee
where
\[
M_n(t) \, = \, (h_{n-1}*\mu)(t) .
\]
For $Tm<s<T(m+1)$, equations \eqref{Proof-Th11-3a} and
\eqref{Proof-Th11-3b} are 
\be \label{Proof-Th11-4a}
X(s) \, = \, \lambda \, \sum^{m}_{j=1} 
\int^s_0 dt \, M_{n}(s-t) \, {\cal G}(t,X(t)) 
\delta ((t+\epsilon)/T-j) ,
\ee
\be \label{Proof-Th11-4b}
P_k(s) \, - \, \sum^{n-k-1}_{j=0} P^{(j)}_k(0) \, h_{j+1} (s) \, 
= \, \lambda \, \sum^{m}_{j=1} 
\int^s_0 dt \, h_{n-k}(s-t) \, {\cal G}(t,X(t)) 
\delta ((t+\epsilon)/T-j) .
\ee
Using equality \eqref{Proof-delta}, equations \eqref{Proof-Th11-4a}, \eqref{Proof-Th11-4b} give
\be \label{Proof-Th11-5a}
X(s) \, = \, \lambda \, T \, \sum^{m}_{j=1} 
M_{n}(s- (Tj-\epsilon)) {\cal G}(Tj-\epsilon,X(Tj-\epsilon)) ,
\ee
\be \label{Proof-Th11-5b}
P_k(s) \, - \, \sum^{n-k-1}_{j=0} P^{(j)}_k(0) \, h_{j+1} (s) \, 
= \, \lambda \, T \, \sum^{m}_{j=1} 
h_{n-k}(s-(Tj-\epsilon)) \, {\cal G}(Tj-\epsilon,X(Tj-\epsilon)) .
\ee
Equations \eqref{Proof-Th11-5a}, \eqref{Proof-Th11-5b} 
give solutions \eqref{EQ-Th11-2a} and \eqref{EQ-Th11-2b}.

%%% Part 2
Using \eqref{Proof-Th11-5a} and \eqref{Proof-Th11-5b}
for $s= T(m+1)-\epsilon$ and
$s= Tm-\epsilon$, we obtain
\be \label{Proof-Th11-6a}
X(T(m+1)-\epsilon) \, = \, \lambda \, T \, \sum^{m}_{j=1} 
M_{n}(T(m+1)- Tj) \, {\cal G}(Tj-\epsilon,X(Tj-\epsilon)) ,
\ee
\[
P_k(T(m+1)-\epsilon) \, - 
\, \sum^{n-k-1}_{j=0} P^{(j)}_k(0) \, h_{j+1} (T(m+1)-\epsilon) \, = \, 
\]
\be \label{Proof-Th11-6b}
\lambda \, T \, \sum^{m}_{j=1} 
h_{n-k}(T(m+1)-Tj) \, {\cal G}(Tj-\epsilon,X(Tj-\epsilon)) ,
\ee
and
\be \label{Proof-Th11-7a}
X(Tm-\epsilon) \, = \, \lambda \, T \, \sum^{m-1}_{j=1} 
M_{n}(Tm- Tj) \, {\cal G}(Tj-\epsilon,X(Tj-\epsilon)) ,
\ee
\be \label{Proof-Th11-7b}
P_k(Tm-\epsilon) \, - \, \, \sum^{n-k-1}_{j=0} P^{(j)}_k(0) \, h_{j+1} (Tm-\epsilon) \, = \,
\lambda \, T \, \sum^{m-1}_{j=1} 
h_{n-k}(Tm-Tj) \, {\cal G}(Tj-\epsilon,X(Tj-\epsilon)) .
\ee
Using variables \eqref{EQ-Th11-XPk}, solutions \eqref{Proof-Th11-6a} and \eqref{Proof-Th11-6b}, \eqref{Proof-Th11-7a} and \eqref{Proof-Th11-7b}, at the limit $\epsilon \to 0+$ give
\be \label{Proof-Th11-8a}
X_{m+1} \, = \, \lambda \, T \, \sum^{m}_{j=1} 
M_{n}(T(m-j+1)) {\cal G}(Tj,X_j) ,
\ee
\be \label{Proof-Th11-8b}
P_{k,m+1} \, - 
\, \, \sum^{n-k-1}_{j=0} P^{(j)}_k(0) \, h_{j+1} (T(m+1)-\epsilon) \, = \,
\lambda \, T \, \sum^{m}_{j=1} 
h_{n-k}(T(m-j+1)) \, {\cal G}(Tj,X_j) ,
\ee
and
\be \label{Proof-Th11-9a}
X_m \, = \, \lambda \, T \, \sum^{m-1}_{j=1} 
M_{n}(T(m-j)) \, {\cal G}(Tj,X_j) ,
\ee
\be \label{Proof-Th11-9b}
P_{k,m} \, - \, 
\sum^{n-k-1}_{j=0} P^{(j)}_k(0) \, h_{j+1} (Tm) \, = \
\lambda \, T \, \sum^{m-1}_{j=1} 
h_{n-k}(T(m-j)) \, {\cal G}(Tj,X_j) .
\ee
Subtracting from equations \eqref{Proof-Th11-8a}, \eqref{Proof-Th11-8b} equation \eqref{Proof-Th11-9a} and \eqref{Proof-Th11-9b}, we obtain 
\[
X_{m+1} - X_m \, = \lambda \, T \, M_{n}(T) {\cal G}(Tm,X_m) +
\]
\be \label{Proof-Th11-10a}
\, \lambda \, T \, \sum^{m-1}_{j=1} 
\Bigl( M_{n}(T(m-j+1)) \, - \, M_{n}(T(m-j)) \Bigr) \, {\cal G}(Tj,X_j) ,
\ee
\[
P_{k,m+1} \, - \, P_{k,m} \, = \,
\sum^{n-k-1}_{j=0} P^{(j)}_k(0) \, 
\Bigl( h_{j+1} (T(m+1)) \, - \, h_{j+1} (Tm) \Bigr) \, + \,
\]
\be \label{Proof-Th11-10b}
\lambda \, T \, h_{n-k}(T) \, {\cal G}(Tm,X_m) \, + \, 
\lambda \, T \, \sum^{m-1}_{j=1} 
\Bigl( h_{n-k}(T(m-j+1)) \, - \, h_{n-k}(T(m-j)) \Bigr)
\, {\cal G}(Tj,X_j) ,
\ee
where $k=1,\dots , n-1$.
Then using \eqref{Omega-10}, equations \eqref{Proof-Th11-10a} 
\eqref{Proof-Th11-10b}, take form \eqref{EQ-Th11-3a}, \eqref{EQ-Th11-3a}.

\end{proof}

%%%%%%%%%%%%%%%%%%%%%%%%%%%%%%%%%%%%%%%%%%%%%%%%%%%%%%%%%%
%%%%%%%%%%%%%%%%%%%%%%%%%%%%%%%%%%%%%%%%%%%%%%%%%%%%%%%%%%
%%%%%%%%%%%%%%%%%%%%%%%%%%%%%%%%%%%%%%%%%%%%%%%%%%%%%%%%%%

Let us derive general nonlocal maps from the equations with the GFD $D^{t,*,n}_{(K)}$ of arbitrary orders and periodic kicks.

\begin{Theorem} \label{Theorem-12}
Let functions $K(t)$, $M(t)$ be a pair of kernels from the Luchko set ${\cal L}_{n,0}$, and $X(t) \in C^m_{-1}(0,\infty)$.
Then the equations
\be \label{EQ-Th12-1a}
D^{t,*,n}_{(K)} [\tau] \, X(\tau) = \lambda \, {\cal G}(t,X(t)) 
\sum^{\infty}_{j=1} \delta ((t+\epsilon)/T-j) ,
\ee
\be \label{EQ-Th12-1b}
D^{t,*,n-k}_{(K)} [\tau] \, P^{*}_k(\tau) = \lambda \, {\cal G}(t,X(t)) 
\sum^{\infty}_{j=1} \delta ((t+\epsilon)/T-j) ,
\ee
where $k=1, \dots , n-1$, have the solutions
\be \label{EQ-Th12-2a}
X(t) \, = \sum^{n-1}_{j=0} X^{(j)}(0)\, h_{j+1}(t) \, + \, 
\lambda \, T \, \sum^{m}_{j=1} 
M(t-Tj+\epsilon) \, {\cal G}(Tj-\epsilon,X(Tj-\epsilon)) ,
\ee
\be \label{EQ-Th12-2b}
P^{*}_k(t) \, = \, \sum^{n-k-1}_{j=0} P^{(j)}(0)\, h_{j+1}(t) 
\, + \, 
\lambda \, T \, \sum^{m}_{j=1} 
M_{n-k}(t-Tj+\epsilon) \, {\cal G}(Tj-\epsilon,X(Tj-\epsilon)) ,
\ee
if $Tm<t<T(m+1)$, where
\be \label{Mn-k}
M_{n-k}(t) \, := \, (h_{n-k-1}*\mu)(t) , \quad
M_{0}(t) \, := \, \mu(t) .
\ee
For the time points $s = Tj-\epsilon$ at $\epsilon \to 0+$, and 
the variables 
\begin{equation} \label{EQ-Th12-XPk}
X_{j}=\lim_{\epsilon \rightarrow 0+} X(Tj-\epsilon), \quad
P^{*}_{k,j}=\lim_{\epsilon \rightarrow 0+} P^{*}_k(Tj-\epsilon),
\quad (j=1, \dots ,m+1) ,
\end{equation} 
solutions \eqref{EQ-Th12-2a}, \eqref{EQ-Th12-2b} are represented by the general nonlocal mappings
\[
X_{m+1} = X_m \,+ \, 
\sum^{n-1}_{j=0} X^{(j)}(0) \, \Omega_{(h_{j+1})}(T,m) \, + \,
\]
\be \label{EQ-Th12-3a}
\lambda \, T \, M_{n}(T) \, {\cal G}(T(m+1),X_m) +
\lambda \, T \, \sum^{n-1}_{j=1} 
\Omega_{(M_{n-k})}(T, m-j) \, {\cal G}(jT,X_j) ,
\ee
\[
P^{*}_{k,m+1} = P^{*}_{k,m} \,+ \,
\sum^{n-k-1}_{j=0} P^{*,(j)}_k (0) \, 
\Omega_{(h_{j+1})}(T,m) \, + \,
\lambda \, T \, M_{n-k}(T) \, {\cal G}(Tm,X_m) \, +
\]
\be \label{EQ-Th12-3b}
\lambda \, T \, \sum^{m-1}_{j=1} 
\Omega_{(M_{n-k})}(T, m-j) \, {\cal G}(jT,X_j) ,
\ee
where $M(t) \in C_{-1}(0,\infty)$ is a function that is associated kernel to the kernel $K(t)$. 
\end{Theorem}

\begin{proof} 
In the proof, we use the second fundamental theorem for GFC of arbitrary order that is proved in \cite{Luchko2021-2}.
This theorem states that if $(M(t),\ K(t))$ is a kernel pair from the Luchko set $\mathcal{L}_{n,0}$ and $X(t) \in C_{-1}^n(0,\infty)$, then
\begin{equation} \label{2FTCn-b}
I^{s,n}_{(M)}[t] \, D^{t,*,n}_{(K)}[\tau] \, X(\tau) = 
X(s) - \sum^{n-1}_{j=0} X^{(j)}(0)\, h_{j+1}(s) .
\end{equation}
Let us consider equations \eqref{EQ-Th12-1a} and \eqref{EQ-Th12-1b} in the form
\be \label{Proof-Th12-0}
D^{t,*,n-k}_{(K)} [\tau] \, P^{*}_k(\tau) = \lambda \, {\cal G}(t,X(t)) 
\sum^{\infty}_{j=1} \delta ((t+\epsilon)/T-j) ,
\ee
where $k=1, \dots , n-1$ and $P^{*}_k(\tau) = X(t)$ for $k=0$.
The action of the integral operator $I^{s,n-k}_{(M)}[t]$ on equation \eqref{Proof-Th12-0} gives
\be \label{Proof-Th12-1}
I^{s,n-k}_{(M)}[t] \, D^{t,*,n-k}_{(K)} [\tau] \, P^{*}_k(\tau) = 
\lambda \, I^{s,n-k}_{(M)}[t] \, {\cal G}(t,X(t)) 
\sum^{\infty}_{j=1} \delta ((t+\epsilon)/T-j) .
\ee
Using the second fundamental theorems \eqref{2FTCn-b} for $D^{t,*,n}_{(K)}$, we get
\be \label{Proof-Th12-2}
P^{*}_k(s) - \sum^{n-k-1}_{j=0} P^{*,(j)}_k(0)\, h_{j+1}(s) 
= \lambda \, I^{s,n-k}_{(M)}[t] \, {\cal G}(t,X(t)) 
\sum^{\infty}_{j=1} \delta ((t+\epsilon)/T-j) .
\ee
Using definition of $I^{s,n-k}_{(M)}$, we write equation \eqref{Proof-Th12-2} as
\be \label{Proof-Th12-3}
P^{*}_k(s) - \sum^{n-k-1}_{j=0} P^{*,(j)}_k(0)\, h_{j+1}(s) 
= \lambda \, 
\int^s_0 dt \, M_{n-k}(s-t) \, {\cal G}(t,X(t)) 
\sum^{\infty}_{j=1} \delta ((t+\epsilon)/T-j) ,
\ee
where $M_{n-k}(t)$ is defined by \eqref{Mn-k}. 
For $Tm<s<T(m+1)$, equation \eqref{Proof-Th12-3} is 
\be \label{Proof-Th12-4}
P^{*}_k(s) - \sum^{n-k-1}_{j=0} P^{*,(j)}_k(0)\, h_{j+1}(s) 
= \lambda \, \sum^{m}_{j=1}
\int^s_0 dt \, M_{n-k}(s-t) \, {\cal G}(t,X(t)) 
\delta ((t+\epsilon)/T-j) ,
\ee
Using equality \eqref{Proof-delta}, equation \eqref{Proof-Th12-4} gives
\be \label{Proof-Th12-5}
P^{*}_k(s) - \sum^{n-k-1}_{j=0} P^{*,(j)}_k(0)\, h_{j+1}(s) 
= \lambda \, T \, \sum^{m}_{j=1}
M_{n-k}(s-Tj+\epsilon) \, {\cal G}(Tj-\epsilon,X(Tj-\epsilon)) ,
\ee
Equation \eqref{Proof-Th12-5} gives solutions \eqref{EQ-Th12-2a} and \eqref{EQ-Th12-2b}.

%%% Part 2
Using \eqref{Proof-Th12-5} for $s= T(m+1)-\epsilon$ and
$s= Tm-\epsilon$, we obtain
\[
P^{*}_k( T(m+1)-\epsilon) - \sum^{n-k-1}_{j=0} P^{*,(j)}_k(0)\, h_{j+1}( T(m+1)-\epsilon) \, = \, 
\]
\be \label{Proof-Th12-6}
\lambda \, T \, \sum^{m}_{j=1}
M_{n-k}( T(m+1)-Tj) \, {\cal G}(Tj-\epsilon,X(Tj-\epsilon)) ,
\ee
\[
P^{*}_k( Tm-\epsilon) - \sum^{n-k-1}_{j=0} P^{*,(j)}_k(0)\, h_{j+1}( Tm-\epsilon) \, = \,
\]
\be \label{Proof-Th12-7}
\lambda \, T \, \sum^{m-1}_{j=1}
M_{n-k}( Tm-Tj) \, {\cal G}(Tj-\epsilon,X(Tj-\epsilon)) .
\ee
Using variables \eqref{EQ-Th12-XPk}, solutions \eqref{Proof-Th12-6} and \eqref{Proof-Th12-7} at the limit $\epsilon \to 0+$ give
\be \label{Proof-Th12-8}
P^{*}_{k,m+1} \, - \,
\sum^{n-k-1}_{j=0} P^{*,(j)}_k(0)\, h_{j+1}( T(m+1)) 
= \lambda \, T \, \sum^{m}_{j=1}
M_{n-k}( T(m-j+1)) \, {\cal G}(Tj,X_j)) ,
\ee
\be \label{Proof-Th12-9}
P^{*}_{k,m} \, - \,
\sum^{n-k-1}_{j=0} P^{*,(j)}_k(0)\, h_{j+1}(Tm) 
= \lambda \, T \, \sum^{m-1}_{j=1}
M_{n-k}( T(m-j)) \, {\cal G}(Tj,X_j) .
\ee
Subtracting equation \eqref{Proof-Th12-9} from equation \eqref{Proof-Th12-8}, we obtain 
\[
P^{*}_{k,m+1} \, - \, P^{*}_{k,m} \, = \,
\sum^{n-k-1}_{j=0} P^{*,(j)}_k(0)\, 
\Bigl( h_{j+1}(T(m+1)) \, - \, h_{j+1}(Tm) \Bigr) \, + \,
\]
\[
\lambda \, T \, M_{n-k}(T) \, {\cal G}(Tm,X_m)) \, + \,
\]
\be \label{Proof-Th12-10}
\lambda \, T \, \sum^{m-1}_{j=1}
\Bigl( M_{n-k}( T(m-j+1)) \, - \, M_{n-k}( T(m-j)) \Bigr) \, 
{\cal G}(Tj,X_j)) ,
\ee
where $k=0,1,\dots , n-1$ with $P_0(t)=X(t)$ .
Then using function \eqref{Omega-10}, equations \eqref{Proof-Th12-10} takes form \eqref{EQ-Th12-3a} and \eqref{EQ-Th12-3b}.

\end{proof}

%%%%%%%%%%%%%%%%%%%%%%%%%%%%%%%%%%%%%%%%%%%%%%%%%%%%%%%%%%
%%%%%%%%%%%%%%%%%%%%%%%%%%%%%%%%%%%%%%%%%%%%%%%%%%%%%%%%%%
%%%%%%%%%%%%%%%%%%%%%%%%%%%%%%%%%%%%%%%%%%%%%%%%%%%%%%%%%%

\section{Conclusion}

In this work, we derive general fractional dynamics with discrete time from general fractional dynamics with continuous time.
Starting from equations with general fractional integral and derivatives with Sonin kernels, we derive exact solutions of these equations. Then, using these solutions for discrete time points, we obtain general universal maps with non-locality in time without approximations.
The universality of the mapping is due to the use of an arbitrary nonlinear function $N={\cal G}(t,X(t))$, and the generality is due to
the general operators kernels $M(t)$ and $K(t)$ of GFC of 
from the Luchko set ${\cal L}_{n,0}$.

We assume that the proposed general nonlocal mappings can be studied by methods and equations proposed in work \cite{Edelman2021}. This possibility is due to the following. 
Note that the proposed general nonlocal maps, which are derived from general fractional differential and integral equations, can be represented by equation \eqref{Proof-Th2-5}, \eqref{Proof-Th3-2c}, \eqref{Proof-Th6-5}, which have the form 
\be \label{Proof-Th2-5b}
X_n = X(0) \, + \, \lambda \, T \, \sum^{n-1}_{k=1} 
M(T(n-k)) \, {\cal G}(Tk,X_k) .
\ee
\be \label{Proof-Th3-5b}
X_n = \, \lambda \, T \, \sum^{n-1}_{k=1} 
M(T(n-k)) \, {\cal G}(Tk,X_k) .
\ee 
\be \label{Proof-Th6-5b}
X_{n} \, = \, \lambda \, T \, \sum^{n-1}_{k=1} 
K^{(1)}(T(n-k)) \, {\cal G}(Tk,X_k) .
\ee
We see that all proposed general universal maps \eqref{Proof-Th2-5b}, \eqref{Proof-Th3-5b}, \eqref{Proof-Th6-5b} with non-locality in time and ${\cal G}(t,X)={\cal G}(X)$, which are derived from equations with GFD and GFI, can be represented as
\be \label{Universal}
X_n = X_0 \, - \, \sum^{n-1}_{k=1} U(n-k) \, G^0(X_k)
\ee
with the function 
\be
G^0(X_k) \, = \, - \, \lambda \, T \, {\cal G}(X_k) ,
\ee
and the kernel
\be
U(n-k) \, = \, M(T(n-k)) ,
\ee
%%%\quad \text{or} \quad
\be
U(n-k) \, = \, K^{(1)}(T(n-k)) ,
\ee
where $M(t)$ and $K(t)$ are the Sonin's kernels of general fractional integral $I^{t}_{(M)}$, and general fractional derivatives $D^{t}_{(K)}$, $D^{t,*}_{(K)}$.

We see that general nonlocal mappings \eqref{Proof-Th2-5b}, \eqref{Proof-Th3-5b}, \eqref{Proof-Th6-5b} are described by a discrete convolution. 
Equations of type \eqref{Universal} is important to study the chaotic and regular behavior of fractional systems that are nonlocal in time. 

Equation \eqref{Universal} coincides with Equation 6 in \cite{Edelman2021}, which is starting point for the study of fractional mappings with non-locality in time, in the framework of the approach proposed by Mark Edelman in the work \cite{Edelman2021}. 

The general nonlocal maps of arbitrary orders, which are derived from equations with GFI and GFDs of arbitrary orders and described by Theorems \eqref{Theorem-10}, 
\eqref{Theorem-11}, \eqref{Theorem-12}, 
can also be represented by equations with discrere convolution.
These maps can be given by equations \eqref{Proof-Th10-7b}, and
\eqref{Proof-Th11-9a}, \eqref{Proof-Th11-9b} and \eqref{Proof-Th12-9} in the form
\be \label{Proof-Th10-7b-Conclusion}
V_{k,m} = \lambda \, T \, \sum^{m-1}_{j=1}
K^{(n-k)}(T(m-j)) \, {\cal G}(Tj,X_j) ,
\ee
where $k=0,1,\dots, n-1$ and $V_{0,m}=X_m$.
\be \label{Proof-Th11-9a-Conclusion}
X_m \, = \, \lambda \, T \, \sum^{m-1}_{j=1} 
M_{n}(T(m-j)) \, {\cal G}(Tj,X_j) ,
\ee
\be \label{Proof-Th11-9b-Conclusion}
P_{k,m} \, - \, 
\sum^{n-k-1}_{j=0} P^{(j)}_k(0) \, h_{j+1} (Tm) \, = \
\lambda \, T \, \sum^{m-1}_{j=1} 
h_{n-k}(T(m-j)) \, {\cal G}(Tj,X_j) .
\ee
where $k=1,\dots, n-1$.
\be \label{Proof-Th12-9-Conclusion}
P^{*}_{k,m} \, - \,
\sum^{n-k-1}_{j=0} P^{*,(j)}_k(0)\, h_{j+1}(Tm) 
= \lambda \, T \, \sum^{m-1}_{j=1}
M_{n-k}( T(m-j)) \, {\cal G}(Tj,X_j) .
\ee
where $k=0,1,\dots, n-1$ and $P^{*}_{0,m}=X_m$.

The chaotic and regular behavior of systems, which are described by such general nonlocal mappings of arbitrary orders, 
can be investigated by generalizing the Edelman method \cite{Edelman2021} to nonlocal mappings of arbitrary order. 

Behavior of fractional dynamical systems with non-locality in time can be very different from the behavior of the dynamical systems with locality in time. 
To study and describe the chaotic and regular behavior of dynamical systems, it is important to know periodic points. 
Fractional dynamical systems have only fixed points, but
these systems can have asymptotically periodic points (sinks) \cite{Edelman2021}. 
For the first time, a method and equations, which allow one to find asymptotically periodic points for nonlinear fractional systems with non-locality in time, were proposed in work \cite{Edelman2021}. 
In article \cite{Edelman2021}, the equations, which can be used to calculate coordinates of the asymptotically periodic sinks, are derived.

%%%%%%%%%%%%%%%%%%%%%%%%%%%%%%%%%%%%%%%%%%%%%%%%%%%%%%%%%%
%%%%%%%%%%%%%%%%%%%%%%%%%%%%%%%%%%%%%%%%%%%%%%%%%%%%%%%%%%
%%%%%%%%%%%%%%%%%%%%%%%%%%%%%%%%%%%%%%%%%%%%%%%%%%%%%%%%%%

%%%In our work builds bridges between two interesting results 
%%%published in 2021 by Yuri Luchko 
%%%\cite{Luchko2021-1,Luchko2021-2} and 
%%%Mark Edelman \cite{Edelman2021}.

%%%%%%%%%%%%%%%%%%%%%%%%%%%%%%%%%%%%%%%%%%%%%%%%%%%%%%%%%%
%%%%%%%%%%%%%%%%%%%%%%%%%%%%%%%%%%%%%%%%%%%%%%%%%%%%%%%%%%
%%%%%%%%%%%%%%%%%%%%%%%%%%%%%%%%%%%%%%%%%%%%%%%%%%%%%%%%%%

%%%%%%%%%%%%%%%%%%%%%%%%%%%%%%%%%%%%%%%%%%%%%%%%%%%%%%%%%%
%%%%%%%%%%%%%%%%%%%%%%%%%%%%%%%%%%%%%%%%%%%%%%%%%%%%%%%%%%
%%%%%%%%%%%%%%%%%%%%%%%%%%%%%%%%%%%%%%%%%%%%%%%%%%%%%%%%%%

\end{document}